\documentclass[pdflatex,sn-mathphys-num]{sn-jnl}% Math and Physical Sciences Numbered Reference Style
\usepackage{graphicx}%
\usepackage{multirow}%
\usepackage{amsmath,amssymb,amsfonts}%
\usepackage{amsthm}%
\usepackage{mathrsfs}%
\usepackage[title]{appendix}%
\usepackage{xcolor}%
\usepackage{textcomp}%
\usepackage{manyfoot}%
\usepackage{booktabs}%
\usepackage{algorithm}%
\usepackage{algorithmicx}%
\usepackage{algpseudocode}%
\usepackage{listings}%
\usepackage{xspace}
\usepackage{pgfplots}
\usepackage{comment}
\usepackage{adjustbox}
\usepackage{tabularx} % put this in your preamble
\usepackage{enumitem}
\usepackage{subcaption}
\definecolor{rootcol}{RGB}{0, 0, 0}
\definecolor{l1col}  {RGB}{55, 126, 184}
\definecolor{l2col}  {RGB}{228, 26, 28}
\definecolor{l3col}  {RGB}{77, 175, 74}
\definecolor{l4col}  {RGB}{152, 78, 163}

\usepackage{cleveref} 
\usepackage[mode=buildnew]{standalone}

\newcommand{\ctoprule}[1]{%
  \cmidrule[\heavyrulewidth]{#1}%
}

\newcommand{\cbottomrule}[1]{%
  \cmidrule[\heavyrulewidth]{#1}%
}

\newcommand\skipvs{skip values\xspace}
\newcommand\identical{interchangeable\xspace}

\newcommand{\rank}{\mbox{\rm\texttt{rank}}}
\newcommand{\select}{\mbox{\rm\texttt{select}}}
\newcommand{\findc}{\mbox{\rm\texttt{find\_close}}}
\newcommand{\findo}{\mbox{\rm\texttt{find\_open}}}

\newcommand{\opp}{\mbox{\rm\texttt{(}}}
\newcommand{\cpp}{\mbox{\rm\texttt{)}}}
\newcommand{\BP}{B}

\newcommand{\CBP}{B_c} 
\newcommand{\Pru}{S_c}
\newcommand{\Tar}{R_c} % reference
\newcommand{\Len}{L_c}

\newcommand{\ignore}[1]{}

\newcommand{\Oh}{\mathcal O}

\theoremstyle{thmstyleone}%
\theoremstyle{thmstyletwo}%

\theoremstyle{thmstylethree}%
\newtheorem{definition}{Definition}%
\newtheorem{lemma}{Lemma} 

\begin{document}

\title[Extended Depth-First Representations of $k^2$-trees]{Extended Depth-First Representations of $k^2$-trees}


%% --- Authors ---

\author*[1]{\fnm{Gabriel} \sur{Carmona}}\email{gabriel.carmona@phd.unipi.it}

\author[1,2]{\fnm{Paolo} \sur{Ferragina}}\email{Paolo.Ferragina@santannapisa.it}

\author[1]{\fnm{Giovanni} \sur{Manzini}}\email{giovanni.manzini@unipi.it}

\author[2]{\fnm{Francesco} \sur{Tosoni}}\email{Francesco.Tosoni@santannapisa.it}

%% --- Affiliations ---

\affil[1]{\orgdiv{Department of Computer Science}, \orgname{University of Pisa}, \orgaddress{\city{Pisa}, \country{Italy}}}

\affil[2]{\orgdiv{L'EMbeDS Department}, \orgname{Sant'Anna School of Advanced Studies}, \orgaddress{\city{Pisa}, \country{Italy}}}


% Purpose - Methods - Results - Conclusion

\abstract{In this paper, we study static, computation-friendly, lossless compression formats for graphs, focusing on memory locality and operational efficiency of $k^2$-trees. We observe that their traditional level-wise layouts suffer from poor cache performance due to weak locality, especially in operations such as matrix-vector and matrix-matrix operations. To address this limitation, we propose four depth-first representations of $k^2$-trees: a plain depth-first layout (EDF-1), a balanced-parenthesis representation (BP), and their compressed variants (CEDF and CBP). We further introduce a linear-time compression method based on suffix and LCP arrays to identify and compress identical subtrees.

We experimentally evaluate the execution time, the disk space, and the peak-memory usage of our approaches against classical level-wise $k^2$-trees and DFUDS-based representations across two real and one synthetic dataset (i.e., Web Graphs, Wikidata, and random adjacency matrices) over the above linear-algebra operations. Results show that our depth-first layouts are competitive and often superior than known approaches: CEDF achieves the best compression in most settings, EDF-1 and CEDF reduce the peak memory usage consistently, and performance varies by workload, with different layouts excelling in different operations and data regimes.

Overall, this work demonstrates that depth-first layouts of $k^2$-trees provide a practical and efficient alternative to traditional layouts, improving both compression and computational performance in matrix operations.
}


\maketitle

% ------------------------------------------
\section{Introduction}
% ------------------------------------------

Graphs are a fundamental abstraction in modern data pipelines, used to model, analyse, and query phenomena across diverse domains~\cite{big-graph-cacm2021}. They provide intuitive, domain-oriented representations, exemplified by Knowledge Graphs (KGs)~\cite{google_kg,amazon_kg,orkg,dblp_lessons}. Prominent instances include Google's Knowledge Graph~\cite{google_kg}, Amazon's Product Graph~\cite{amazon_kg}, the Open Research Knowledge Graph (ORKG)~\cite{orkg}, Wikidata~\cite{wikidata-survey,wikidata-making} (via Wikimedia Enterprise~\cite{wikimedia_enterprise_wikidata_api_2026}), and DBLP~\cite{dblp_lessons,kg_role_dagstuhl}, to mention a few. Consequently, graph technologies have become more than ever before a cornerstone of data innovation, particularly through Graph Databases (GDBs)~\cite{AnglesG08,LiangMTZ24}. The global GDB market, valued at USD 2.85 billion in 2025, is projected to grow to USD 20.29 billion by 2034\footnote{\url{https://www.fortunebusinessinsights.com/graph-database-market-105916}}. The ISO standardisation of the GQL query language~\cite{GQL} further signals their growing significance. Additionally, Retrieval-Augmented Generation (RAG) models, which structure knowledge as KGs to overcome Large Language Model (LLM) limitations~\cite{graph_rag}, have accelerated this trend. 

Motivated by the ubiquity of massive graphs and their applications~\cite{big-graph-cacm2021}, this paper investigates static, computation-friendly, losslessly compressed graph formats. Existing literature includes WebGraph~\cite{webgraph_java,webgraph_rust}, Zuckerli~\cite{zuckerli}, mm-repair~\cite{improving_gramcompr_mat_pvldb}, and the $k^2$-tree~\cite{k2-tree_udc,brisaboa2009k2,navarro_book}. These compressed formats reduce storage requirements while supporting operations directly over the compressed data, ranging from neighbour retrieval to algebraic operations over the graph adjacency matrix, such as matrix-vector~\cite{j4,Francisco2022} or matrix-matrix~\cite{k2-tree_vldbj} multiplication. Heightened interest in graph processing, driven by AI inference advances, has spurred research also to other compressed formats, such as Compressed Linear Algebra (CLA)~\cite{cla_vldbj,cla_cacm} and beyond~\cite{gram_la_neurips2020,accelerating_gnn2025,k2-tree_vldbj,aware,uni_lcompr_sparse,improving_gramcompr_mat_pvldb,efficient_compact_via_entropy}. These methods typically treat matrix rows as sequences, often employing reordering strategies to enhance compression over non-binary matrices.

Among all those approaches, $k^2$-trees~\cite{k2-tree_udc,brisaboa2009k2}\cite[\S9.2.1]{navarro_book} are particularly interesting when the matrices to compress are binary, sparse, and with some clustering features of the input data. A $k^2$-tree is built on a binary matrix whose size $n$ is a power of $k$, possibly padded with zeros. It is a $k^2$-ary tree, in which the root represents the entire $n \times n$ matrix, and its children partition it into $n^2/k^2$ square submatrices. A bit `0' denotes a submatrix full of 0s, and it becomes a leaf of the $k^2$-tree; while a bit `1' denotes a non-null submatrix, whose node is further decomposed. This yields a tree of height $1 + \log_k n$. In this paper, we assume $k=2$, as is common in literature and available implementations.

\begin{figure}[tbp]
  \centering
  \includegraphics[width=.4\textwidth]{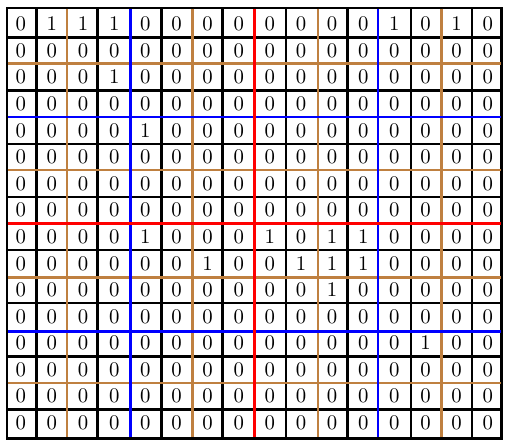}
  \includegraphics[width=\textwidth]{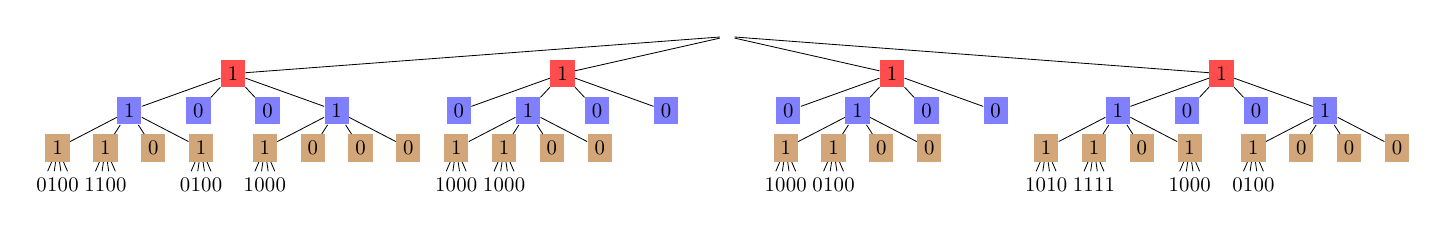}
  \vspace{-.1cm}
    \caption{A sparse binary matrix and its associated $k^2$-tree ($k=2$).}
    \label{fig:matrix-example}
\end{figure}

\Cref{fig:matrix-example} illustrates a $16 \times 16$ matrix partitioned into $2 \times 2$ submatrices (top) and its corresponding $k^2$-tree (bottom). Submatrix dimensions are colour-coded: higher levels represent larger submatrices. Although depicted with pointers, the $k^2$-tree is pointerless; navigation relies on traversing the bit vector encoding internal nodes in row-major order.

The $k^2$-tree efficiently supports operations on its compressed bit vector, such as retrieving a matrix element, a column, or a row, and performing other matrix operations, such as logical operations, addition, and multiplication, directly on their compressed form~\cite {k2-tree_brisaboa, k2-tree_vldbj,6824442}. These operations are essential to support more sophisticated queries over the corresponding input graph, such as Regular Path Queries (RPQs), which extend the expressiveness of query languages like Cypher~\cite{cypher}, GQL~\cite{GQL}, PGQL~\cite{PGQL}, TigerGraph~\cite{tiger_graph}, and SPARQL 1.1~\cite{modern_query_lang_graph_db}. Notably, 24\% of the 208 million queries in the Wikidata Query Service logs involve at least one RPQ feature~\cite{wd_query_log,the_most_out_wd}. Additionally, PageRank kernels~\cite{j4,Francisco2022} are efficiently implemented via $k^2$-trees. These features made the $k^2$-tree a very well-known approach to the implementation of the adjacency matrix of graphs, such as Web Graphs~\cite{k2-tree_udc,claude2010fast}, Graph DBs~\cite{k2-tree_vldbj,6824442}, Geographic Information Systems~\cite{DEBERNARDO202386,k2-raster}, and Social Networks~\cite{DEBERNARDO202386}. 

%Whereas two-dimensional Block Trees~\cite{2DBT} have extended their structure to discrete (i.e., non-binary), repetitive data.

However, the traditional level-by-level layout of $k^2$-trees does not make good use of locality: nodes that are close in the tree may be laid out far apart in memory. For operations that need to inspect most of the tree (such as matrix multiplication, matrix addition, or matrix-vector multiplication), this lack of locality often leads to many cache misses. A natural solution to address this problem is to use a depth-first layout, which places the nodes of the same subtree close together in memory. Following this idea, we propose two depth-first representations of $k^2$-trees: one based on reordering the level-wise binary sequence in depth-first order, and one based on the balanced parenthesis encoding of trees. Our representations have another advantage: their depth-first layout makes it easier to detect identical subtrees, which can then be identified and compressed, thus originating two other (compressed) representations. 

As a final result, our paper introduces four new depth-first representations of $k^2$-trees: the plain depth-first layout (EDF-1), the balanced-parenthesis layout (BP), and their compressed versions (CEDF and CBP). We also present a linear-time method, based on Suffix and LCP arrays, to detect and compress repeated subtrees. At the same time as the conference version of this paper \cite{conf_version}, Fariña et al. \cite{cache-friendly-boolean} also
introduced a cache-friendly $k^2$-tree layout inspired by the DFUDS representation, which is also depth-first but different from ours.

We experimentally evaluate our proposals against two established baselines: the classical level-wise $k^2$-tree~\cite{k2-tree_vldbj} and the DFUDS-based $k^2$-tree of Fariña et al.~\cite{cache-friendly-boolean}. The experiments cover three families of binary matrices of different sizes and sparsity levels, either derived from known large graphs (such as Web Graph~\cite{datasets1,datasets2}, Wikidata database~\cite{k2-tree_vldbj}), or generated randomly~\cite {k2-tree_vldbj}. We measure disk space usage, execution time, and peak memory during matrix-vector multiplication, matrix-matrix addition, and matrix-matrix multiplication.

Overall, our depth-first representations, both plain and compressed, result competitive and often advantageous. The CEDF layout achieves the best disk space usage in most datasets, while EDF-1 and CEDF consistently use the least memory. Performance varies by operation and dataset: EDF-1 is generally the fastest for matrix-vector multiplication, the DFUDS-based layout performs best for matrix-matrix addition, and no single method dominates in matrix-matrix multiplication. EDF-1 and CEDF tend to perform better on denser matrices or when many subtree traversals occur, whereas the classical level-wise layout excels on very sparse matrices. CBP representation exploits repeated subtrees in some WebGraph datasets, but it is generally slower and, in the case of BP, requires more space than all other solutions.

As a final note, this paper extends the conference version that appeared in~\cite{conf_version} by providing an expanded introduction, new sections, a theoretical analysis of the space and time complexity of our four new depth-first representations (summarized in Lemma~\ref{lemma:bfs},~\ref{lemma:edf},~\ref{lemma:bp},~\ref{lemma:dfuds},~\ref{lemma:cbp}, and~\ref{lemma:cedf}), additional experimental comparisons (matrix-vector multiplication and matrix-matrix addition), and improved implementations and thus final results.

The remainder of this paper is organised as follows.~\Cref{sec:notation} introduces notation and terminology.~\Cref{sec:canonical} reviews canonical $k^2$-trees.~\Cref{sec:df-traversal} describes our new and known depth-first layouts (Plain, Enriched, Balanced Parenthesis, and DFUDS).~\Cref{sec:compression} presents our subtree-compression techniques.~\Cref{sec:experiments} reports experimental results. Finally,~\Cref{sec:conclusion} concludes the paper.

% ------------------------------------------
\section{Notation and terminology}
\label{sec:notation}
% ------------------------------------------

Let $\Sigma$ be a finite ordered alphabet of constant size $\sigma$. 
A {\em string} (or sequence or array) of length $n$ over an alphabet $\Sigma$ is denoted with $S[1,n] \in \Sigma^n$. We write $S[i..j]$ to denote the substring $S[i]S[i+1]\cdots S[j]$, if $1\leq i \leq j \leq n$, or the empty string otherwise. The Suffix Array~\cite{MM93} of $S$ is a permutation of the integers $\{1,\ldots,n\}$ such that for $i=2,\ldots,n$, $S[SA[i-1],n] \prec S[SA[i],n]$, where $\prec$ denotes the lexicographic ordering. The LCP Array~\cite{MM93}\cite[\S10.2.2]{pearls_ferragina} of string $S$ is an array of integers such that, for  $i=2,\ldots,n$, $LCP[i]$ is the length of the longest common prefix between $S[SA[i-1],n]$ and $S[SA[i],n]$. Both $SA$ and $LCP$ can be computed in optimal $O(n)$ time~\cite{KSB06,karkkainen2009permuted,KoAlu03,KSPP03}.

Given a string $S$ and a symbol $c\in\Sigma$,  $\rank_c(S,i)$ returns the number of occurrences of $c$ in $S[1,i]$, and $\select_c(S,i)$ returns the position of the $i$-th occurrence of $c$ in $S$ or $-1$ if such occurrence does not exist. Using additional $o(n)$ bits of space we can preprocess $S$ so that both $\rank$ and $\select$ operations can be computed in $O(1)$ time~\cite[\S4.3]{navarro_book}\cite{Munro96}. If $\Sigma = \{ \opp, \cpp \} $, the sequence $S$ is called a {\em balanced parenthesis} sequence~\cite[\S8.2]{navarro_book} if $S$ contains the same number of \opp\ and  \cpp\ symbols and no prefix $S[1,i]$ contains more $\cpp$ than $\opp$ then $S$. Given a balanced parenthesis sequence $S$, $\findc(S, i)$ (resp. $\findo(S, i)$) returns the index of the closing (opening) parenthesis that matches a given opening (closing) parenthesis $S[i]$. Such operations can be supported in $O(1)$ time~\cite{Munro_Raman_2002}.

% ------------------------------------------
\section{The canonical representation of \texorpdfstring{$k^2$}~-trees}
\label{sec:canonical}
% ------------------------------------------

Given a square binary matrix $M$, the corresponding $k^2$-tree~\cite{brisaboa2009k2} is built by recursively dividing $M$ using an \textit{MX-Quadtree} strategy~\cite{samet2006foundations}, splitting it into $k^2$ equal-sized submatrices, each containing $n^2/k^2$ cells. Each submatrix corresponds to a child of the root node, with its value set to $1$ if at least one cell in the submatrix is $1$; otherwise, its value is $0$. The process is recursively applied to all submatrices with at least one $1$ until all cells have been processed, as illustrated in~\Cref{fig:matrix-example}. 
This approach works seamlessly when $n$ is a power of $k$. In cases where $n$ is not a power of $k$, the matrix can be padded with additional rows and columns of zeros until its size becomes a power of $k$. In what follows, we assume $k$ is a constant, so the size of the padded matrix remains $O(n)$.

By construction, the $k^2$-tree is structured as a $k^2$-ary tree of height $h=\lceil \log_k n\rceil$, where each node, except the root, stores a single bit of information. For all levels $\ell<h$, internal nodes store a $1$ while leaf nodes store a $0$, meaning that the corresponding $k^{h-\ell}\times k^{h-\ell}$ submatrix only contains zero elements. 
At level $h$, all nodes are leaves, each storing a bit value from the input matrix. 
In~\cite{brisaboa2009k2}, the authors propose a representation of the above tree (and thus of the associated Boolean matrix) consisting of two bit vectors: $T$ ({\em tree}) stores all the bits that are not in the last level of the tree following a top-down level-wise traversal. In contrast, $L$ ({\em leaves}) stores all the bits of the last level of the tree in left-to-right order. In the following, we refer to the above as the {\em canonical representation} for the $k^2$-tree. For instance, for the $k^2$-tree in~\Cref{fig:matrix-example}, the corresponding $T$ and $L$ bit vectors are as follows (tree levels are colour-coded following~\Cref{fig:matrix-example}):
\begin{itemize}[leftmargin=20mm]
    \item[\small{$T$:}] \small{\texttt{{\color{red}1111} {\color{blue}1001 0100 0100 1001} {\color{brown}1101 1000 1100 1100 1101 1000}}}
    \item[\small{$L$:}] \small{\texttt{0100 1100 0100 1000 1000 1000 1000 0100 1010 1111 1000 0100}}
\end{itemize}
There is a reason to keep the bit vectors $T$ and $L$ separately. To access specific elements of the input matrix, it is necessary to navigate the $k^2$-tree efficiently. The authors observed that we can compute the $i$-th child of a node at position $x$ using the formula $\text{\texttt{child}}(x, i) = \rank_1(x)\cdot k^2+i$~\cite{brisaboa2009k2}. Hence, by adding an $o(|T|)$-bit data structure that supports the \texttt{rank$_1$} operation over $T$ in constant time, it is possible to navigate the tree efficiently. 
Later, Brisaboa et al.~\cite{k2-tree_udc} extended the functionalities of the $k^2$-tree with operations still relying only on the \rank\ operation on~$T$. Conversely, there is no need to enrich $L$ with additional data structures, so it is stored as a plain bitvector.

For an $n\times n$ matrix with $m$ nonzeros,~\cite{brisaboa2009k2} showed that the number of bits in this representation is bounded by $k^2m\left(\log_{k^2} ({n^2}/{m}) \, +\, O(1)\right)$. Yet, this value is obtained only for pathological inputs; in practice, experiments over web graphs~\cite{k2-tree_udc} showed that the actual space usage is significantly lower and competitive with other graph compression schemes.

We summarise the space and traversal of a subtree time complexity of the canonical $k^2$-tree with Lemma~\ref{lemma:bfs}.

\begin{lemma}\label{lemma:bfs}
    A  $k^2$-tree that has $t$ nodes at levels $0,\dots,h-1$ and $\ell$ nodes at level $h$ can be represented using the Enriched Depth-First format in $t+\ell+ o((t+\ell))$ bits and supports the traversal of a subtree in constant time.
\end{lemma}

% ------------------------------------------
\section{Depth-First Traversal Representations}
\label{sec:df-traversal}
% ------------------------------------------

The canonical representation of the $k^2$-tree is very space-efficient, using a single bit per node, while still supporting tree navigation in constant time. Yet, despite this substantial efficiency, such a representation poses challenges for identifying the possible presence of identical submatrices in the input matrix. Indeed, such submatrices (corresponding to identical subtrees), due to the breadth-first layout of the canonical representation, have their bits scattered across different tree levels and appear far apart from each other in the {\em tree} $T$ bit sequence. The layout thus makes it a nontrivial task to detect identical subtrees. In addition, the canonical representation does not favour the locality of accesses: nodes close to each other in the tree may be stored far apart in $T$; thus, for certain computations, such as subtree extractions, the layout can lead to many cache misses. To improve locality of accesses and mitigate to the above shortcomings, in the remainder of this section, we describe alternative representations whose layout is informed by a {\em depth-first} traversal of the $k^2$-tree: such an approach presents the substantial advantage of storing the information of a given subtree together in a contiguous portion of the bit sequence, thus improving the locality of accesses and simplifying the task of detecting identical subtrees.

% ------------------------------------------
\subsection{Plain Depth-First representation}
\label{subsec:plain-df}
% ------------------------------------------

The simplest Depth-First layout exhibits node descriptors in the order in which they appear, traversing the tree in depth-first order. To construct this representation, when the DF-visit reaches a non-leaf node $u$, we write the $k^2$ bits associated with $u$'s children. The resulting bit vector $P$ has length $N k^2$, where $N$ is the number of internal nodes of the $k^2$-tree; see~\Cref{fig:k2tree-pdft} for an example.
It is easy to notice that $P$ is a permutation of the bits in $T \cup L$ of the canonical representation. More precisely, both $P$ and $T \cup L$ can be partitioned into $N$ blocks of $k^2$ bits (i.e., $4$ bits when $k=2$). The blocks of $P$ are a permutation of the blocks of $T \cup L$.

\ignore{The simplest Depth-First layout exhibits node descriptors in the order in which they appear, traversing the tree in depth-first order. To construct this representation, when the DF-visit reaches a non-leaf node $u$, we write the $k^2$ bits associated with $u$'s children. The resulting bit vector $P$ has length $m k^2$, where $m$ is the number of internal nodes of the $k^2$-tree; see~\Cref{fig:k2tree-pdft} for an example.\footnote{Recall that $m$ is the number of non-zero values in the binary matrix over which the $k^2$-tree is built. And this number equals the number of internal nodes of the $k^2$-tree.}  
It is easy to notice that $P$ is a permutation of the bits in $T \cup L$ of the canonical representation. More precisely, both $P$ and $T \cup L$ can be partitioned into $m$ blocks of $k^2$ bits (i.e., 4 bits when $k=2$). The blocks of $P$ are a permutation of the blocks of $T \cup L$.}

\begin{figure}[t]
    \centering
    \includegraphics[width=1.0\textwidth]{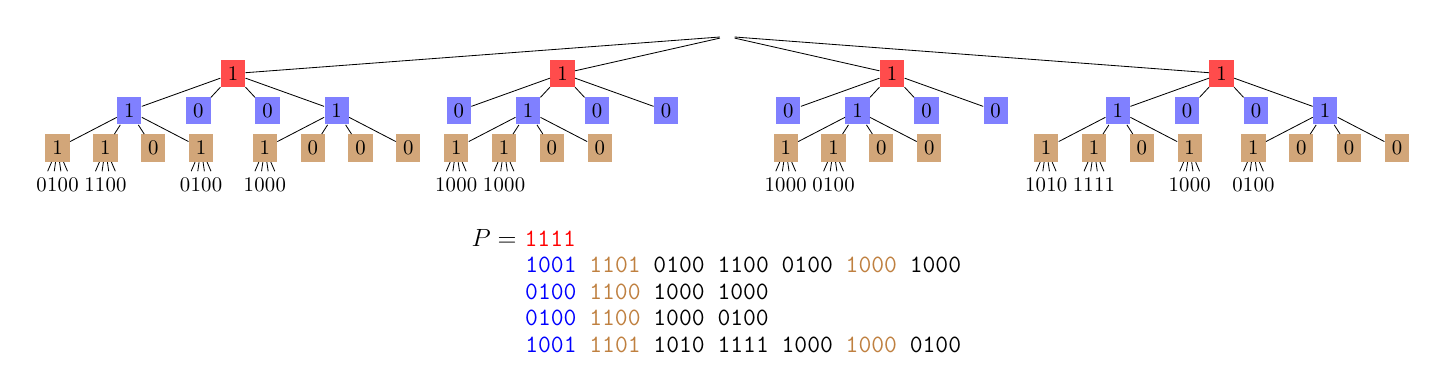}
    %\includestandalone[width=.7\textwidth]{fig/k2tree-pdft.tex}
    \caption{The Plain Depth-First representation of the same tree as~\Cref{fig:matrix-example}. The representation consists of the bit vector $P$, which contains the concatenation of the 4-tuples of bits shown above; the first row represents the children of the root node. Each of the subsequent four rows corresponds to the subtrees rooted at each child of the root, listed in depth-first traversal order; different colours represent different tree levels.}
    \label{fig:k2tree-pdft}
\end{figure}

We call this representation the Plain Depth-First (PDF) representation since we do not add any additional information to speed up the tree navigation. Nonetheless, given a pair of indices $i,j$, we are still able to determine the entry $M[i][j]$ of the matrix represented by the $k^2$-tree, even though this access operation requires a depth-first traversal of the tree up to the node representing the smallest non-empty submatrix containing position $i,j$. Though access to a single entry is inefficient, operations involving the whole matrix (such as computing the vector product $y = M x$) that require visiting the whole $k^2$-tree can be carried out in time proportional to the size of the compressed representation via a simple left-to-right scan of the bit vector $P$.  Due to subtrees appearing in contiguous portions of $P$, the DF-visit is more cache-friendly than visiting the canonical BF-based representation.

% ------------------------------------------
\subsection{Enriched Depth-First representation}
\label{subsec:enriched-df}
% ------------------------------------------

The main drawback of the PDF representation of~\Cref{subsec:plain-df} is arguably that finding the starting position in $P$ of the subtree corresponding to the, say, third child of the root node, requires performing a visit of the first two subtrees. The cost of the visit is proportional to the number of nodes in such subtrees, since the visit consists of a linear scan of a subarray of $P$. So this is a major problem when such subtrees are large. 

The above observation suggests enriching the plain depth-first representation with additional information that allows us to skip large subtrees without visiting them. The simplest approach is to associate with the block of $k^2$ bits representing an internal node the sizes of its subtrees, except for the last one. In the example in~\Cref{fig:k2tree-pdft}, with the block {\textcolor{red}{\tt 1111}}, we would store the sizes of the first three subtrees, e.g., 7 (first subtree), 4 (second subtree), and 4 (third subtree). If we need to access the third subtree, we skip $7+4=11$ blocks, and to access the fourth subtree, we skip $7+4+4=15$ blocks. We do not store the size of the last subtree (7 in our example) because it is not used to skip any subtree. 

Storing this information for all nodes would be too expensive, so we choose a threshold $\tau$ and store the ``skip'' values only for subtrees of size larger than~$\tau$. For example, if we set $\tau=6$, we use ``skip'' values only for the root and for the first and last children of the root, whose subtrees have size 7. Since those children have two subtrees, they only need to store a single skip value (4 in our example). If $N$ is the total number of nodes (remember that $n$ is the matrix size) and we choose $\tau = f(N)$, we can estimate the overhead of storing the \skipvs as follows. Consider the subtrees that have more than $f(N)$ nodes and do not contain any subtree with more than $f(N)$ nodes. Clearly, there are at most $N/f(N)$ such subtrees, and each such subtree will have at most $\log_k n$ ancestors. Not all ancestors are distinct, but we can state that the number of nodes containing \skipvs is at most $O(N\log n /f(N))$.  Assuming we use $O(\log N)$ bits for each skip value, the total overhead for skip values is $O(N k^2\log N\log n /f(N))$ bits. In our experiments, we set $\tau = \sqrt{N}$ and $k=4$ so the overhead is $O(\sqrt{N}\log N\log n)$ bits. Since $N=|P|+1$, this is $o(|P|)$ when $N=\Omega(\log^{2+\epsilon} n)$, that is when the matrix is not pathologically sparse.

We call the above representation the Enriched Depth-First (EDF). A few details of the representation must be decided during implementation. The first is how to encode the skip values: to save space, one could use a variable-length code (e.g., Elias, Rice, Variable-Byte; cf.~\cite[\S11]{pearls_ferragina}) depending on the desired time-space trade-off. In our experiments, we found that when we set $\tau = \Theta(\sqrt{N})$, the number of skip values is relatively small, so we store them explicitly by using $\log N$ bits. We indeed posit that further compressing them would incur a considerable slowdown in exchange for modest compression gains.

The second implementation choice is {where} to store the \skipvs. A first alternative is to interleave them in the bit vector $P$. This approach is cache-friendly because the entire representation remains within a single bit vector, which is mostly accessed via sequential scans. However, the size of the subtrees, and hence the skip values, now also depend on the skip values stored at the lower levels: this makes the construction of the representation more complex, since it now has to be done bottom-up and left-to-right. We therefore decided to store the \skipvs in a separate vector $S$. This is also non-trivial since, to skip a subtree $T_i$, we need to skip the portion of the $P$ vector containing the encoding of $T_i$ nodes {\em and} the portion of the $S$ vector containing the \skipvs for $T_i$. Hence, the size of this portion of $S$ must also be stored; this can be done within the vector $S$ itself. In matrix-matrix multiplication, where submatrices are repeatedly traversed, we combine the above {\em static} set of \skipvs with a {\em dynamic} set: when the recursive multiplication algorithm reaches a submatrix whose corresponding subtree has no \skipvs, such values are computed on the fly for all nodes in the subtree and stored in a temporary vector $S'$. Since, by construction, the subtrees with no \skipvs have fewer than $\tau$ nodes, computing $S'$ takes $O(\tau)$ time and space. 

We summarise the result of this section as follows:

\begin{lemma}\label{lemma:edf}
    A  $k^2$-tree that has $t$ nodes at levels $0,\dots,h-1$ and $\ell$ nodes at level $h$ can be represented using the Enriched Depth-First format in $t+\ell+ O((t+\ell)k^2\log (t+\ell) \log n / f(t+\ell))$ bits and supports the traversal of a subtree in $O(\tau)$ time, with $\tau=f(t+\ell)$.
\end{lemma}

%---------------------------------------------------

% ------------------------------------------
\subsection{Balanced Parenthesis representations}
\label{subsec:bp}
% ------------------------------------------

In this section, we introduce a new DF-based representation of $k^2$-trees, based on Balanced Parenthesis (BP) encoding, following the ideas of Munro and Raman~\cite{Munro_Raman_2002}. 
Our representation follows the classical approach used for general trees; yet, we exploit the special structure of the $k^2$-tree, introducing an optimisation for the last-level nodes. Instead of representing these nodes with parentheses, we store their actual values (i.e., bits) in a separate binary vector~$L'$. We describe two implementations of this idea, each using a different auxiliary data structure to support tree navigation. 

% ------------------------------------------
\subsubsection{Classical Balanced Parenthesis representation}
\label{subsec:BPclas}
% ------------------------------------------

Let $h = \lceil\log_k n\rceil$ denote the tree height. 
The construction of our BP representation is done as follows. We visit the $k^2$-tree in depth first order, and we write to a vector $\BP$ an open parenthesis \texttt{(} every time we start the visit of a node, and we write a closing parenthesis \texttt{)} every time the visit of a node is complete, and we go back to the parent node. However, when we reach a level-$(h-1)$ node that is not a leaf, instead of visiting its children, we write a pair \texttt{()} and write the $k^2$ bits associated with its children (i.e., values of those leaves) to the bit vector $L'$. See~\Cref{fig:k2tree-bp} for an example. 

Note that the subtrees rooted at a level $h-1$ are represented by either \texttt{()}, if they have no children, or by \texttt{(())} if they do have children. Since in a $k^2$-tree every node has either $k^2$ children or none, the sequence \texttt{(())} cannot represent a subtree rooted at a level $< h-1$. Hence, there is a one-to-one correspondence between the occurrences of the pattern 
\texttt{(())} in $\BP$, the subtrees with children at level $h-1$, and the blocks of $k^2$ bits in the $L'$ vector. By construction, the correspondence is order-preserving in the sense that the $i$-th occurrence of \texttt{(())} corresponds to the $i$-th block in $L'$.
This means that to access the bit values stored in a level-$h$ leaf, we need to locate the position $p$ of its parent, and count the number $x$ of occurrences of the pattern  \texttt{(())} in $\BP$ up to position $p$; the desired bit values are stored starting from position $x k^2$ in $L'$.  

\begin{figure}[t]
    \centering
    %\includestandalone[width=\textwidth]{PAPER-SPIRE/graphics/k2tree-bp.tex}
    % \includestandalone[width=\textwidth]{fig/k2tree-bp.tex}
    \includegraphics[width=\textwidth]{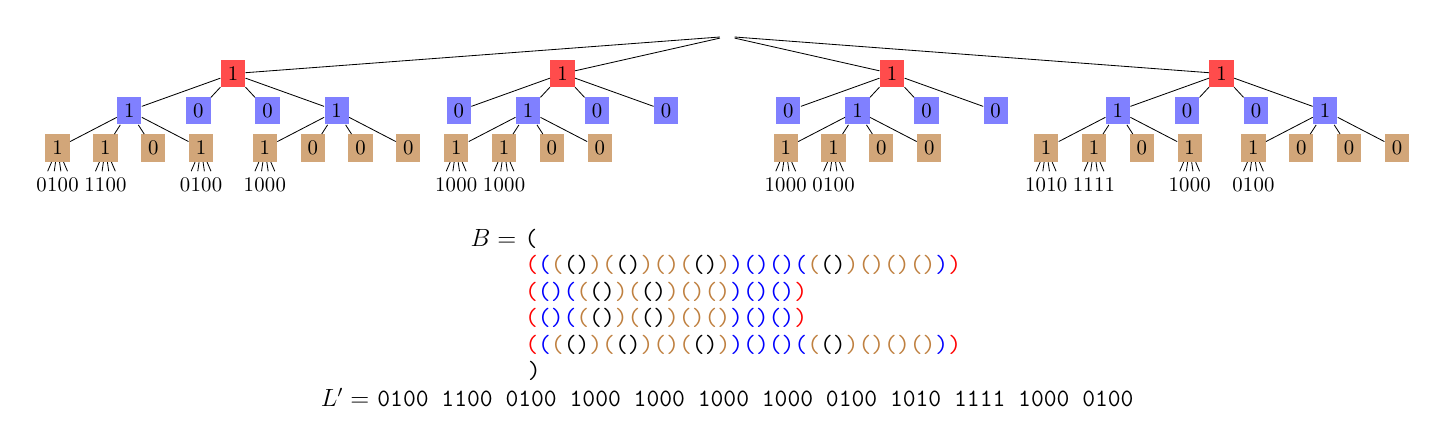}
    \caption{The balanced parenthesis representation of the $k^2$-tree of~\Cref{fig:matrix-example}. The resulting parenthesis vector $\BP$ is the concatenation of six rows of parentheses:
    The first and last row are the open and closed parentheses for the root, and the other four rows represent the subtrees rooted at the four children of the root. Also shown is the bit vector $L'$, which contains the values at the bottom level, stored in depth-first visit order. Notice that $L'$ consists of $12$ groups of $k^2=4$ bits, and $B$ contains indeed $12$ occurrences of the substring \texttt{(())}.}
    \label{fig:k2tree-bp}
\end{figure}

Let $t$ denote the number of nodes in levels $0,\ldots, h-1$ and $\ell$ denotes the number of level-$h$ leaves. 
Our representation uses one bit for each leaf, two parentheses for each internal node, and an additional \texttt{()} pair for each one of the $\ell/k^2$ non-leaf nodes at level $h-1$. Hence the parenthesis vector $\BP$ and the bit vector $L'$ have length

\begin{equation}\label{eq:bp}
|\BP| = 2t +\frac{2\ell}{k^2},\qquad
|L'| = \ell.     
\end{equation}

The BP representation, therefore, takes roughly twice as much space as the canonical and (Enriched) Depth-First representations, which take $t+\ell$ bits plus possibly lower-order terms to support fast navigation; see Lemmas~\ref{lemma:bfs} and~\ref{lemma:edf}.

The advantages of the BP representation are locality of reference and the possibility, as explored in~\Cref{sec:compression}, to compress identical subtrees. To ensure constant-time navigation with the BP representation, we need to support the \findc\ operation in constant time, which, as recalled in~\Cref{sec:notation}, can be done using $o(|\BP|)$ bits of auxiliary space.  In addition, when we reach an internal node at level $h-1$, to read the corresponding bits in~$L'$ we need to support the constant time \rank\ operation for the pattern $\texttt{(())}$ on the vector $\BP$; this can be achieved by a straightforward modification of the \rank\ data structure, still using $o(|\BP|)$ bits of auxiliary space. 

% ------------------------------------------
\subsubsection{Full traversal BP variant}
\label{subsec:BPeff}
% ------------------------------------------

The main drawback of the previous approach is that using an additional data structure to support the \texttt{find\_close} operation in constant time introduces extra overhead, on top of the already doubled space required by the \textit{BP} representation. Yet, this is necessary if we want a ``fully functional'' $k^2$-tree that achieves the same time complexity for all operations, such as random access. Since this paper focuses on operations that require a full traversal of the tree, it is natural to look for alternative solutions.

To replace the data structure used for traversal, we propose a two-level solution that combines the enriched nodes introduced in Section~\ref{subsec:enriched-df} with excess array techniques (inspired by~\cite{Munro01, navarro2014fully, GEARY2006231}).

At the first level, we enrich the nodes that are roots of subtrees larger than a threshold $\tau$. The same analysis from~\Cref{subsec:enriched-df} applies to the space overhead. Even though this was originally studied for EDF, it allows us to skip at least $\tau$ nodes; in terms of bits, this is significant, since we can skip at least $2 \cdot \tau$ bits.

The second level consists of the classical excess precomputation~\cite[\S7]{navarro_book} over the tree. The excess corresponds to the number of opening parentheses minus the number of closing parentheses in the range $B[1, i]$. 

Now, we divide the sequence into blocks of size $e$. For each block, we store: the excess at the end of the block, the minimum excess within the block, and the number of leaves in the block. The first two values are standard in solutions for computing \texttt{find\_close} over a sequence of parentheses. The third value is specific to the $k^2$-tree and avoids the need for a \rank\ operation over the pattern \texttt{(())}. We can exploit the particular structure of the $k^2$-tree to store all this information efficiently. The excess can be stored using $\log h = \log \log n$ bits, since the height of the tree bounds it. The minimum excess of a block can also be stored in $\log \log n$ bits. Finally, the number of leaves can be stored in $\log \frac{e}{4}$ bits, because at most $e/4$ leaves can appear within a block of $e$ parentheses.

In terms of space complexity, we still require $o(|B|)$ bits of auxiliary space, but in practice this leads to lower overall space usage. As for time complexity, the most expensive operations in this layout are those requiring skipping a subtree, such as random access. The time complexity of the \texttt{find\_close} operation becomes $\Oh (\tau / e)$. If $\tau = \Oh (\sqrt{|B|})$, then this is $\Oh (\sqrt{|B|} / e)$. This relation holds since, in the worst case, we traverse the tree in constant-time steps until reaching a subtree of size $\tau - 1$. If we then need to perform \texttt{find\_close}, we must scan that subtree using the block structure described above, which requires $(\tau - 1) / e$ steps.

In practice, we contend that the requirement to scan the subtrees does not significantly affect performance, as it can be carried out efficiently given that the subtrees are small. In operations such as matrix-matrix multiplication, there is a high probability of traversing most of the tree anyway. Indeed, we will see in~\Cref{sec:experiments} that this technique considerably improves space usage, both in memory and on disk, as well as execution time, compared to the classical BP representation (as implemented in our original paper~\cite{conf_version}).

We summarise the result of this section as follows:

\begin{lemma}\label{lemma:bp}
    A  $k^2$-tree that has $t$ nodes at levels $0,\dots,h-1$ and $l$ nodes at level $h$ can be represented using the Balanced Parenthesis format in $2t+2\ell/k^2+\ell+ o(t+\ell/k^2)$ bits and supports the traversal of a subtree in $\Oh(\tau)$ time.
\end{lemma}

\subsection{DFUDS representation}
\label{subsec:dfuds}
%---------------------------------------------------

Simultaneously to the proposal of our depth-first representations~\cite{conf_version}, Fari\~na et al.\ introduced in~\cite{cache-friendly-boolean} a DFUDS-based representation. Let $t$ be the number of nodes in levels $0,\ldots,h-1$, and let $\ell$ be the number of level-$h$ leaves. If the representation directly used DFUDS, the space required would be $2(t+\ell)$ bits. To avoid this overhead, they propose a representation that retrieves DFUDS information from the $k^2$-tree signatures.

They build the $k^2$-tree using a $DF$-visit, as we do in the PDF representation. Then, for each group of $k^2$ bits, they prepend one bit: $0$ if the node is not at the last level, and $1$ if the node is at level $h-1$. With this signature, they can differentiate internal nodes from value matrices. Using this bit, when reading a node, they transform the $k^2$ bits into their DFUDS representation. If the signature is $0$, the DFUDS representation of that node is $(^c)$, where $c$ is the number of children. If the signature is $1$, the DFUDS representation is $(^c))^c$, which corresponds to the $c$ children at the last level plus the $c$ values at the last level.

\begin{figure}[t]
\centering
\includegraphics[width=\textwidth]{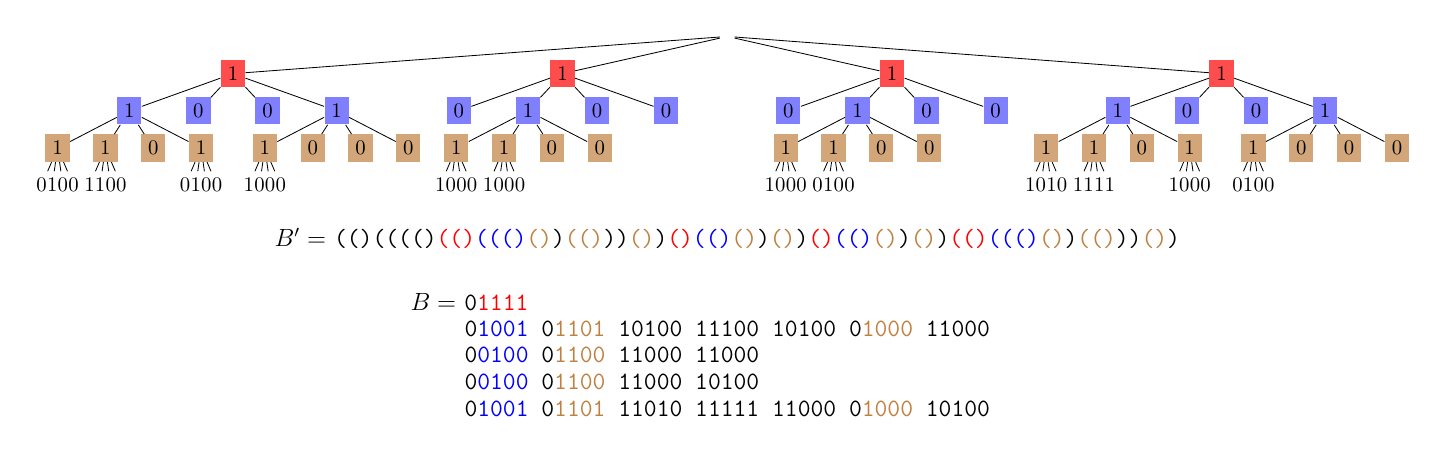}
\caption{The cardinal DFUDS representation of the $k^2$-tree of~\Cref{fig:matrix-example} and the DFUDS representation by Fari\~na et al.~\cite{cache-friendly-boolean}. The initial sequence {\tt (()} is prepended as a guard and is therefore ignored in the sequence $B$. Note that $B'$ is not stored; it is only shown for illustration.}
\label{fig:k2tree-dfuds}
\end{figure}

\Cref{fig:k2tree-dfuds} shows $B'$, the DFUDS corresponding to the $k^2$-tree representation described above. As an example, the first node \texttt{0\color{blue}1001} corresponds to the first node \texttt{\color{red}(()} in $B$, and the first node \texttt{101000} corresponds to the node \texttt{{\color{brown}()})} in $B$.

With this representation, the number of added bits is $(t+\ell)/k^2$, so the total space used is $t + \ell + (t+\ell)/k^2$ bits.

Since they can transform the nodes of the bit array representing the $k^2$-tree into the DFUDS representation, they show that it is possible to build a \textit{range min-max tree} over the resulting sequence, and to compute \texttt{fwd\_search}, which in DFUDS supports the \texttt{child} operation. The \texttt{fwd\_search} can be computed in constant time, and therefore \texttt{child} can also be computed in constant time. This data structure adds an extra $o(|B|)$ bits.

We summarise the results of this section as follows:

\begin{lemma}\label{lemma:dfuds}
    A  $k^2$-tree that has $t$ nodes at levels $0,\dots,h-1$ and $l$ nodes at level $h$ can be represented using the Depth-First Unary Degree format in $t + \ell + (t+\ell)/k^2+ o(t + \ell + (t+\ell)/k^2)$ bits, and supports the traversal of a subtree in constant time.
\end{lemma}

% ------------------------------------------
\section{Subtree compression on $k^2$-trees}
\label{sec:compression}
% ------------------------------------------

Identical submatrices in the original input matrix $M$ can lead to identical subtrees in the $k^2$-tree representing $M$. We now show how to detect identical subtrees, replace them with a ``pointer'' to a previous occurrence, and perform navigation operations on this compressed representation. The algorithm for detecting identical subtrees is based on the Suffix Array and the LCP array and runs in linear time, i.e., time proportional to the size of the $k^2$-tree representation. 

% ------------------------------------------
\subsection{Compressed Balanced Parenthesis}
\label{subsec:comp-bp} 
% ------------------------------------------

For the BP representations, instead of considering identical subtrees in the traditional sense, we consider the wider concept of \identical subtrees according to the following definition, that, for our convenience, excludes the pathological case of subtrees consisting of a single node. 

\begin{definition} \label{def:isubtrees}
    Let $T_1$ and $T_2$ be subtrees of the $k^2$-tree $T$ featuring height greater than 1. We say that $T_1$ and $T_2$ are {\em \identical} if their respective balanced parenthesis sequences generated during the visit described in~\Cref{subsec:bp} are identical.     
\end{definition}

We point out that the balanced-parenthesis sequences generated during the visit in~\Cref{subsec:bp} do not include any information about the leaves at the last level: we store instead in binary form the leaves representation in the vector $L'$. Hence, some subtrees are \identical even if some of the values stored on the last level differ; see, for example, the subtrees surrounded in red in~\Cref{fig:ex-idem-subtree}.

\begin{figure}[t]
    \centering
    \includegraphics[width=\textwidth]{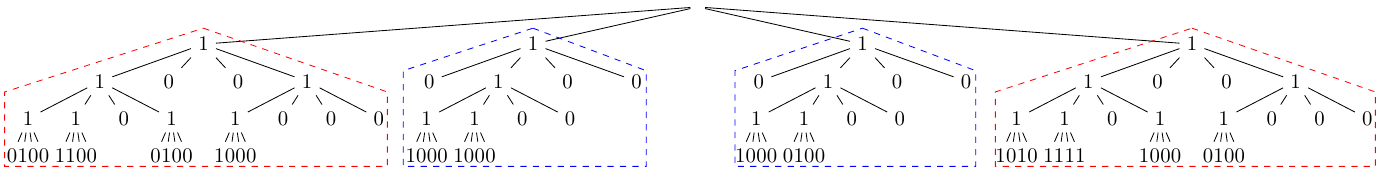}
    \caption{The subtrees surrounded by dashed lines of the same colour are \identical according to Definition~\ref{def:isubtrees}. Note that the subtrees surrounded in red are \identical even if the values stored in some of the leaves at the last level are different.}
    \label{fig:ex-idem-subtree}
\end{figure}

\begin{lemma}\label{lemma:identical}
    The subtree starting at position $i$ in $\BP$ is \identical with the subtree starting in position $j$, if and only if, setting $c=\mbox{\rm \texttt{find\_close}}(i)$, the substring $\BP[i..c]$ is equal to $\BP[j..j+c-i]$. 
\end{lemma}
\begin{proof}
Note that, by construction, it is $\BP[i]=\BP[j]=\texttt{(}$. The position $c=\texttt{find\_close}(i)$ is the position of the closed parenthesis matching $\BP[i]$. Hence, the visit of the entire subtree starting at position $i$ generates the string $\BP[i..c]$. If such a string is identical to $\BP[j..j+c-i]$, then the two subtrees are interchangeable.
\end{proof}

Because of the above Lemma, we can detect all \identical subtrees in linear time as follows. Firstly, we compute the suffix array $SA$ and the $LCP$ array of the parenthesis sequence $\BP$. For $i=1,\ldots |\BP|-1$, if $\BP[SA[i]] = \texttt{(}$, we set $c=\findc(SA[i])$. Then, if $c-SA[i] \leq  LCP[i]$, by the above Lemma we can conclude that the subtree starting at position $SA[i]$ is \identical with the subtree starting at $SA[i-1]$. During this procedure, we ignore positions $i$ when $LCP[i] \leq 32$ (or some other larger threshold) since in that case an \identical subtree would be so small that replacing it with a pointer would offer no benefit.

The above algorithm will find {\em all} \identical subtrees, but for our purposes, we do not need all this potential. The problem is that subtrees of \identical subtrees are still \identical. Yet, we are only interested in finding {\em maximal} \identical subtrees, i.e., subtrees which are not contained in larger \identical subtrees. Thus, instead of scanning the $SA$ array, we scan the parenthesis sequence $\BP$: for $j=1,\ldots,|\BP|$, if $\BP[j] = \texttt{(}$ we set $i=SA^{-1}[j]$ and we check whether $LCP[i] > \texttt{find\_close}(j)-j$. Note that if we have $k$ consecutive suffixes $SA[i],\ldots, SA[i+k-1]$ each representing \identical subtrees, then they are all \identical among themselves. In this case, we select the subtree with the leftmost starting position as the reference subtree, and the other subtrees will point to it.
If the subtree starting at $j$ is \identical with a previous subtree, we restart the scanning of $\BP$ at position $\texttt{find\_close}(j)+1$, thus skipping all non-maximal \identical subtrees contained in the subtree starting at $j$.

Having identified which pointers should replace (maximal) subtrees, we call them {\em pruned} subtrees. The construction of the Compressed Balanced Parenthesis (CBP) representation is done as follows. 
Let $t_c$ be the number of nodes in levels $0,\dots, h-1$, $h$ be the number of tree levels, $\ell$ be the number of level $h$ leaves in $\BP$, $\ell_c$ be the number of leaves at level $h$ remaining in the $\CBP$ after subtree compression, and $p$ be the number of pruned subtrees. Given the uncompressed sequence of balanced parentheses $\BP$, we build the compressed sequence $\CBP$ by scanning $\BP$ from left to right. When we reach a position $i$ such that $\BP[i] = \texttt{(}$ and the subtree rooted at $i$ must be pruned, we stop copying until the matching parenthesis at $\findc(i)$. Then, we write in $\CBP$ a ``special node'' that represents this pruned subtree.

A first option for representing this ``special node'' and, thus, a pruned subtree is to use the pattern \texttt{(())}. Note that $\CBP$ is still balanced. However, the same pattern is also used for internal nodes at level $h-1$. To distinguish both cases, we introduce a binary vector $\Pru[1,\ell_c + p]$ such that: For each occurrence of the pattern \texttt{(())} in $\CBP$, $\Pru[i]$ stores a $0$ if it corresponds to an internal node at level $h-1$; otherwise, it stores a $1$ if it corresponds to a pruned subtree.

We also use an integer array $\Tar[1,p]$, where $\Tar[v]$ stores the starting position in $\CBP$ of the reference subtree for the $v$-th pruned subtree. When navigation reaches a position $i$ in $\CBP$ with an \texttt{(())} pattern that is not at level $h-1$, we compute $r = \rank_{\texttt{(())}}(\CBP, i)$ and then $v = \rank_1(\Pru, r)$. By construction, the pruned subtree at position $i$ is the $v$-th pruned subtree, so $j = \Tar[v]$ gives the starting position of its reference subtree in $\CBP$. To compute the rank of the pattern \texttt{(())}, we use a classical rank technique~\cite{navarro_book}. We precompute the prefix sums over blocks of fixed size (in our experiments, $4096$ bits). Then, within each block, we count the number of occurrences using bitwise operations on machine words of $64$ or $128$ bits. This auxiliary data structure requires an additional $o(|B_c|)$ bits, and since both the block size and the bitwise operations run in constant time, the overall rank computation is performed in constant time.

Since we prune subtrees but not the values in $L'$, we still need to know how many values in $L'$ correspond to each pruned subtree. This can be stored in an integer array $\Len[1,p]$, where $\Len[v]$ corresponds to the $v$-th pruned subtree, or computed on the fly during traversal.

The size of the resulting compressed balanced parentheses is:
$$|\CBP| = 2t_c + 4p + \frac{2\ell_c}{k^2}.$$

In the worst case, when nothing is pruned, we have $t_c = t$ and $p = 0$, so $|\CBP| = |\BP|$.

For the auxiliary structures, we use $(\ell_c + p) + o(\ell_c + p)$ bits for $\Pru$ and for supporting $\rank$ over it. If $p$ is small, a sparse bit-vector representation such as Elias--Fano may be used, requiring
$2p + p \lceil \log((\ell_c + p)/p)\rceil$ bits.
The array $\Tar$ uses $p \log(2t_c + 4p + 2\ell_c/k^2)$ bits.~\Cref{fig:repl-k2tree} shows the $k^2$-tree of~\Cref{fig:k2tree-bp} with two subtrees pruned, and the resulting $\CBP$ sequence.

This first approach is intuitive, but it can use too much space due to the bit vector $\Pru$ when $\ell_c$ and $p$ are large. It also adds time overhead, since we must perform an extra \rank\  operation to check whether an \texttt{(())} pattern corresponds to a pruned subtree or to a leaf. For this reason, we instead represent pruned subtrees using the pattern \texttt{(()())}. This pattern is useful because it cannot appear in a $k^2$-tree subtree, and it lets us remove $\Pru$ completely. We still need to perform a \rank\ over $\CBP$, but now for the pattern \texttt{(()())}. Note that this \rank\ operation is required no matter whether we use \texttt{(())} or \texttt{(()())}. With this change:
$$|\CBP| = 2t_c + 6p + \frac{2\ell_c}{k^2}.$$

This removes $\Pru$ from the space usage and avoids the extra \rank\ operation.

Finally, to support $\findc$, the straightforward solution is to adopt the data structure described in~\Cref{subsec:BPeff}, which yet requires adding extra information on both levels. At the first level, one would also store, along with the enriched values, the number of pruned nodes in the subtree associated with it. In the second level, we will add, for each block, how many times the pattern \texttt{(()())} appears; those values can be encoded using $\log (e/6)$ bits.

\begin{figure}[t]
    \centering
    \includegraphics[width=\textwidth]{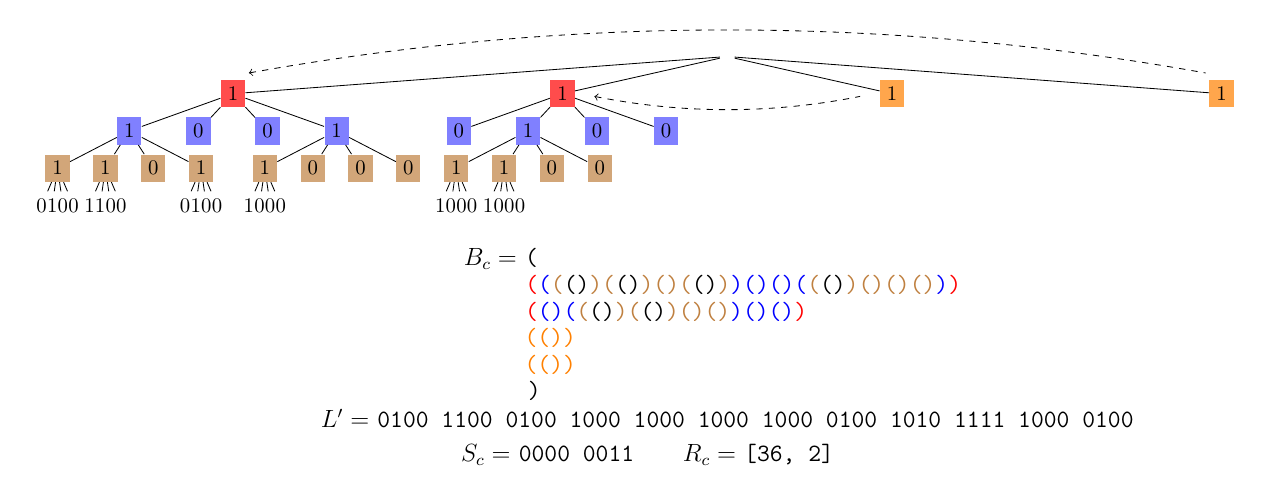}
    %\includestandalone[width=\textwidth]{fig/k2tree-replaced.tex}
    \caption{The Compressed Balanced Parenthesis representation for the $k^2$-tree shown in~\Cref{fig:k2tree-bp} with two subtrees (in orange) pruned and replaced with pointers to \identical subtrees. The two pointers are explicitly stored in the array $\Tar$.}\label{fig:repl-k2tree}
\end{figure}

We summarise the result of this section as follows:

\begin{lemma}\label{lemma:cbp}
    Given a $k^2$-tree of height $h$ with $t$ nodes at levels $0,\dots, h-1$ and $\ell$ leaves at the last level. The respective subtree-compressed $k^2$-tree, that has $t_c$ nodes at levels $0,\dots, h-1$, $l_c$ leaves at level $h$, and $p$ pruned subtrees, can be represented using balanced parentheses and be stored using $2t_c+6p+2\ell_c/k^2+\ell+o(t_c+p+\ell_c/k^2)$ bits. It supports the traversal of a subtree in $O(\tau)$ time.
\end{lemma}

% ------------------------------------------
\subsection{Compressed Plain/Enriched Depth-First}
\label{subsec:comp-df}
% ------------------------------------------

Finding repeated subtrees in the Plain Depth-First (PDF) representation is conceptually similar to the BP case. The main difference is that in PDF, leaf values are mixed with internal nodes; hence, when compressing such representations, we are interested in finding {\em maximal identical subtrees}, i.e., subtrees that match completely from their root to their leaves.

Let $P$ be the binary sequence that represents the PDF encoding of the input $k^2$-tree, assuming $k = 2$. To find identical subtrees, we still use the $SA$ and $LCP$ arrays. However, a substring $P[a,b]$ that encodes a subtree $t$ can match another substring $P[c,d]$ that does not represent a subtree, or does not represent the same subtree. In this situation, the subtree $t$ cannot be compressed (that is, a pointer cannot replace it).~\Cref{fig:counterexamplepdf} shows an example of two different $k^2$-trees that have identical PDF representation, thus illustrating how this problem can occur. It is important to note that this also occurs in the canonical representation, and it is a direct consequence of using only 1 bit per node in the $k^2$-tree. With the $BP$ representation, this issue does not occur because the open and closed parentheses allow us to identify the structure of the tree clearly.

\begin{figure}[t]
\begin{subfigure}{1\textwidth}
  \centering
  \includegraphics[width=\textwidth]{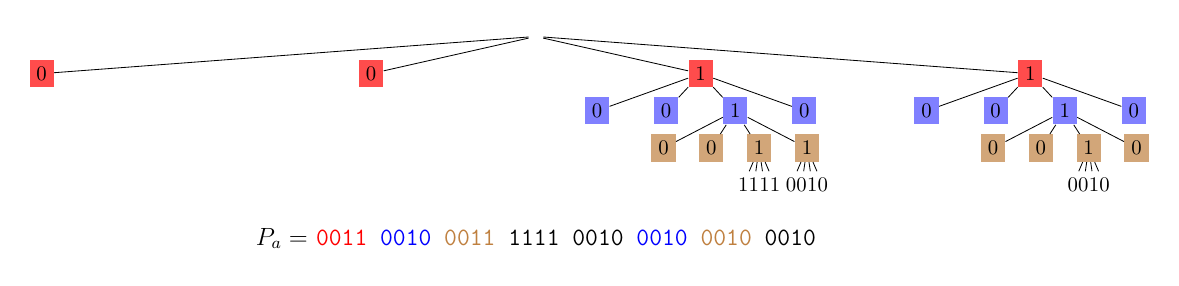}
  \caption{Example of a valid $k^2$-tree for a matrix of size $16\times 16$ and $k=2$.}
  \label{fig:k2pdfa}
\end{subfigure}%

\begin{subfigure}{1\textwidth}
  \centering
  \includegraphics[width=\textwidth]{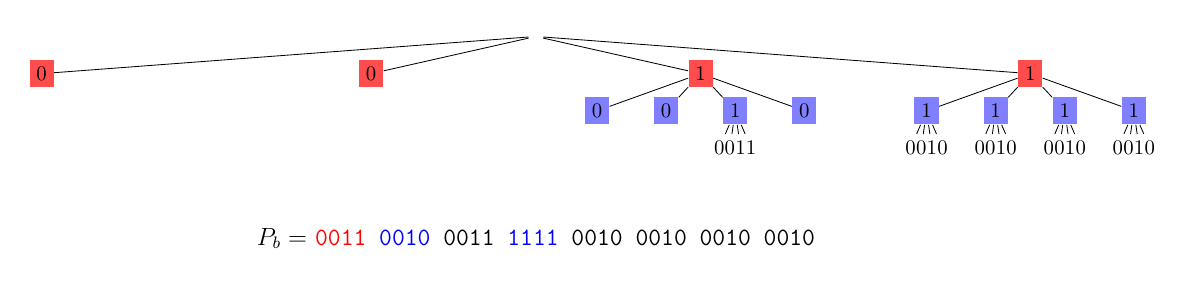}
  \caption{Example of a valid $k^2$-tree for a matrix of size $8\times 8$ and $k=2$.}
  \label{fig:k2pdfb}
\end{subfigure}
\caption{Example of two different $k^2$-tree that have the same PDF representation.}
\label{fig:counterexamplepdf}
\end{figure}

To avoid these spurious matches and detect identical subtrees, we build an auxiliary sequence $P'$, such that $P'[i]$ is a pair containing the depth of node $i$ and the 4-bit block associated with node $i$. For example, for the tree encoding in~\Cref{fig:k2tree-pdft} the sequence $P'$ starts with: {\small{
\texttt{({\color{red}1 1111}) 
({\color{blue} 2 1001}) 
({\color{brown} 3 1101}) 
(4 0100) 
(4 1100) (4 0100) ({\color{brown} 3 1000})}\, \ldots}}

Using the $SA$ and $LCP$ array of $P'$, we can compute the {\em maximal} identical subtrees as in~\Cref{subsec:comp-bp}, by finding equal subsequences $P'[i,i+k]$ and $P'[j,j+k]$  of length at least equal to the number of nodes in the subtree rooted at $P'[i]$. Each subsequence representing a pruned subtree is replaced by the 4-bit block {\tt 0000}, which is a block that cannot appear in the original sequence $P$. As for the BP case, we maintain an integer array $R_c$ storing pointers to the subtree identical to the pruned one. When we reach a node associated with a {\tt 0000} block, to find the corresponding subtree, we only need to count (using a rank data structure) the number of {\tt 0000} blocks up to the current one and read the corresponding pointer from $R_c$ (there is no need for the bit vector $S_c$ since {\tt 0000} can only denote a pruned subtree). 

We call this representation Compressed PDF (CPDF, for short). As in~\Cref{subsec:enriched-df}, we can enrich it by computing the \skipvs on the pruned tree, treating the roots of pruned subtrees as leaves. We store this data separately to speed up navigation in the upper tree levels. The resulting structure is called Compressed EDF (CEDF).

\begin{lemma}\label{lemma:cedf}

    Given a $k^2$-tree of height $h$ with $t$ nodes at levels $0,\dots, h-1$ and $\ell$ leaves at the last level. The respective subtree compressed $k^2$-tree, that has $t_c$ nodes at levels $0,\dots, h-1$, $l_c$ leaves at level $h$, and $p$ pruned subtrees, can be represented using balanced parentheses and be stored using 
    $t_c+4p+\ell_c+ o(t_c+p+\ell_c)$ bits. It supports the traversal of a subtree in $O(\tau)$ time.
\end{lemma}

% ------------------------------------------
\section{Experimental evaluation}
\label{sec:experiments}
% ------------------------------------------

To support our theoretical results, we tested all previous data structures (both existing ones and our new proposals) on three sets of binary matrices that differ in type, size, and sparsity. We refer to these datasets as WebGraph, Database, and Random. Overall, we measured disk space usage, execution time and peak memory usage (PMU). The execution time and peak memory usage were measured using the command \texttt{/usr/bin/time}. Experiments were run on an Ubuntu machine with 394 GiB RAM and an Intel Xeon Gold 6132 @ 2.60 GHz.

% ------------------------------------------
\subsection{Datasets}
% ------------------------------------------

\begin{description}
    \item[WebGraph.] This dataset contains matrices extracted from the WebGraph framework~\cite{datasets1,datasets2}, available at \url{https://sparse.tamu.edu/LAW}, and commonly used in related work~\cite{2DBT,j4}. From this collection, we selected seven graphs: \texttt{CNR-2000} (\texttt{CN-2000}), \texttt{AMAZON-2008} (\texttt{AM-2008}), \texttt{EU-2005}, \texttt{IN-2004}, \texttt{INDOCHINA-2004} (\texttt{ID-2008}), \texttt{ARABIC-2005} (\texttt{AR-2008}), and \texttt{UK-2005}.~\Cref{tab:info-webgraph} reports their main characteristics, including matrix size, number of nonzeros, $k^2$-tree (canonical representation) nodes per nonzero, and density.

    \item[Database.] This dataset contains matrices derived from the Wikidata graph database~\cite{k2-tree_vldbj}. The graph is composed of triples of the form entity-relation-entity, which can be represented as a matrix in which cell $(i,j)$ is set to $1$ if, and only if, entity $i$ is connected to entity $j$ via that relation. For instance, considering the ``occupation'' relation, the entity for ``Federico Fellini'' is connected to the entity for ``film director''\footnote{\url{https://www.wikidata.org/wiki/Q7371}}. We observe that these matrices are generally larger and sparser than those from WebGraph. For our experiments, we retained only matrices with a density of at least $10^{-11}$, which we consider the minimum required for meaningful compression in this paper.~\Cref{tab:info-database} reports which matrices were considered and their main characteristics, including matrix size, number of nonzeros, nodes per nonzero, and density.

    \item[Random.] Following the approach of Arroyuelo et al.~\cite{k2-tree_vldbj}, we generated synthetic matrices of size $1{,}000$, and of density equal to $2\cdot 10^{-1}$, $10^{-1}$, $10^{-2}$, $10^{-3}$, and $10^{-4}$. For each density, we generate $10$ different matrices, given a total amount of $50$ matrices of size $1{,}000$. We remark that the $1$s in these matrices are distributed uniformly.

\end{description}

\begin{table}[t]
    \centering
    \caption{Size, number of nonzeros (edges), $k^2$-tree nodes (canonical representation) per nonzero, and density of nonzero elements for the matrices in the WebGraph dataset. The matrices are ordered by matrix size from left to right.}\label{tab:info-webgraph}
    \begin{tabular}{lrrcr}
       \toprule
       File~~ & Matrix size~~ & \#Nonzeros~~ & Nodes/NZs & Density~~\\
       \midrule
       \texttt{CN-2000} & 325,557 & 3,216,152 & 0.87 & $3\cdot 10^{-5}$\\
       % \hline
       \texttt{AM-2008} & 735,323 & 5,158,388 & 2.32 & $9\cdot 10^{-6}$\\
       % \hline
       \texttt{EU-2005} & 862,664 & 19,235,140 & 1.02 & $2\cdot 10^{-5}$\\
       % \hline
       \texttt{IN-2004} & 1,382,908 & 16,917,053 & 0.73 & $8\cdot 10^{-6}$\\
       % \hline
       \texttt{ID-2004} & 7,414,866 & 194,109,311 & 0.60 & $3\cdot 10^{-6}$\\
       % \hline
       \texttt{AR-2005} & 22,744,080 & 639,999,458 & 0.68 & $1\cdot 10^{-6}$\\
       % \hline
       \texttt{UK-2005} & 39,459,925 & 936,364,282 & 0.68 & $6\cdot 10^{-7}$\\
       \bottomrule
    \end{tabular}
\end{table}

%\begin{figure}[h]
%    \centering
%    \includegraphics[scale=.7]{results/database/info.pdf}
%    \caption{Number of nonzeros, $k^2$-tree nodes (canonical representation) per nonzero, and density of nonzero elements for the matrices used in the Graph Database dataset. The matrices are ordered by how were paired for matrix-matrix operations ($M_1 + M_2$ and $M_1 \cdot M_2$), begin top table $M_1$ and bottom table $M_2$.
 %   \label{tab:info-database}}
%\end{figure}

\begin{table}[htbp]
    \centering
    \caption{Number of nonzeros, $k^2$-tree nodes (canonical representation) per nonzero, and density of nonzero elements for the matrices used in the Graph Database dataset. The matrices are ordered by how were paired for matrix-matrix operations ($M_1 + M_2$ and $M_1 \cdot M_2$), begin top table $M_1$ and bottom table $M_2$.
    \label{tab:info-database}}
    \begin{tabular}{lrrr}
        \toprule
       Matrix ID~ & \# Nonzero~~ & Nodes/Nonz & Density~~~\\
       \midrule
       \texttt{0001} & 105,901,917 & 2.80 & $9\cdot 10^{-10}$\\
       \texttt{0004} & 53,883,532 & 8.79 & $4\cdot 10^{-10}$\\
       \texttt{0006} & 53,883,271 & 6.48 & $4\cdot 10^{-10}$\\
       \texttt{0007} & 53,883,271 & 3.71 & $4\cdot 10^{-10}$\\
       \texttt{0008} & 39,609,445 & 10.29 & $3\cdot 10^{-10}$\\
       \texttt{0012} & 49,516,226 & 4.46 & $4\cdot 10^{-10}$\\
       \texttt{0159} & 10,089,284 & 4.54 & $8\cdot 10^{-11}$\\
       \texttt{0606} & 19,961,090 & 10.56 & $2\cdot 10^{-10}$\\
       \texttt{1619} & 19,430,075 & 9.99 & $2\cdot 10^{-10}$\\
       \bottomrule
        & & & \\
        \toprule
        Matrix ID~ & \# Nonzero~~ & Nodes/Nonz & Density~~~ \\
       \midrule
         \texttt{2831} & 78,222,146 & 10.25 & $6\cdot 10^{-10}$\\
         \texttt{3744} & 18,499,841 & 12.52 & $2\cdot 10^{-10}$\\
         \texttt{0889} & 19,341,438 & 11.52 & $2\cdot 10^{-10}$\\
         \texttt{3868} & 14,895,138 & 10.58 & $1\cdot 10^{-10}$  \\
         \texttt{2670} & 166,682,725 & 8.93 & $1\cdot 10^{-9}$ \\
         \texttt{3936} & 16,611,881 & 6.71 & $1\cdot 10^{-10}$ \\
         \texttt{4660} & 17,457,608 & 11.21 & $1\cdot 10^{-10}$ \\
         \texttt{3867} & 14,895,138 & 10.54 & $1\cdot 10^{-10}$ \\
         \texttt{3935} & 18,421,701 & 8.58 & $2\cdot 10^{-10}$ \\
       \bottomrule
    \end{tabular}  
\end{table}

% ------------------------------------------
\subsection{Matrix operations}
% ------------------------------------------

On the different implementations of the $k^2$-trees, we tested the following operations: 

\begin{description}
    \item[Matrix-vector multiplication.] For each matrix in every dataset, we generated a random double-precision vector and multiplied it by the matrix. We report the average execution time over $100$ executions. In this case, we do not report the peak memory usage (PMU) because this operation doesn't require any additional memory; thus, the PMU equals the space of the $k^2$-tree plus the size of the double-precision vector.\\

    \item[Matrix-matrix sum.] For WebGraph matrices, we computed the sum of each matrix and its transpose. For Database matrices, we paired matrices, computed the sum for each pair, and reported the execution time and the PMU. For Random matrices, as with Database, we paired matrices and computed their sum, but report only the average execution time.\\

    \item[Matrix-matrix multiplication.] The algorithm implementing this operation consists of the classical divide-and-conquer recursive procedure: Given two matrices $A$ and $B$, we partition them into four submatrices of equal size
$$
A = \scalebox{0.8}{$\begin{pmatrix}
 A_0 & A_1\\
 A_2 &  A_3
\end{pmatrix}$}, \quad
B = \scalebox{0.8}{$\begin{pmatrix}
 B_0 &  B_1\\
 B_2 &  B_3
\end{pmatrix}$},$$
\noindent then, their product is obtained by the recursive formula:
$$
A \times B = \scalebox{0.8}{$\begin{pmatrix}
  A_0 \times B_0 + A_1 \times B_2 &  A_0 \times B_1 + A_1 \times B_3 \\
  A_2 \times B_0 + A_3 \times B_2 &  A_2 \times B_1 + A_3 \times B_3
\end{pmatrix}$}.
$$
This approach aligns well with $k^2$-trees for $k=2$, where each submatrix maps to one of the 4 children of a node. For WebGraph matrices, we computed the square of each matrix. For Database matrices, we randomly paired the matrices, creating $9$ pairs, and multiplied them. We will report the execution time and peak memory usage for each pair. For Random matrices, we also paired matrices, this time generating $90$ pairs and multiplying them. Still, since they are random matrices, we will report the average execution.\\ 

\ignore{\item[Matrix transitive closure.] The transitive closure of a matrix is implemented as follows: given a matrix $A$, we compute $A \leftarrow A + A \cdot A$. For WebGraph matrices, we did not test this operation because it has no clear application. For Database matrices, this operation is computationally expensive, so we will report the time required to compute the transitive closure of each matrix. For Random matrices, this operation is faster, and so we will take the average over the $10$ available matrices.}

\end{description}

For the Random dataset we don't measure the PMU, because of three main reasons: (1) the matrices are considerably small; (2) because the elements are uniformly distributed, the compression is limited; (3) the main objective of the Random dataset is to evaluate the time execution by varying the density of 1s.

The results of each operation are discussed in the following Sections~\ref{sec:mv-mul}--\ref{sec:mm-mul}.

% ------------------------------------------
\subsection{Implementations}\label{sec:implementation}
% ------------------------------------------

On each dataset, we tested $2$ known variants of the $k^2$-tree:

\begin{itemize}
    \item \texttt{K2-TREE}\footnote{\href{https://github.com/Yhatoh/rpq-matrix}{\texttt{K2-TREE}: https://github.com/Yhatoh/rpq-matrix}}: the canonical level-by-level $k^2$-tree implementation in C by Arroyuelo et al.~\cite{k2-tree_vldbj} (see~\Cref{sec:canonical}). 
    
    \item \texttt{K2-DFUDS}\footnote{\href{https://github.com/Yhatoh/dfs-kdtree}{\texttt{K2-DFUDS}: https://github.com/Yhatoh/dfs-kdtree}}: an implementation based on the depth-first representation described in~\Cref{subsec:dfuds}, we use $b=1024$ as in the original paper. The original implementation in C++ by Fariña et al.~\cite{cache-friendly-boolean} did not include matrix-vector multiplication, so we implemented it. 
\end{itemize}

And our four new proposals (two plain and two compressed):

\begin{itemize}
    
    \item \texttt{K2-EDF-1} and \texttt{K2-CEDF}\footnote{\href{https://github.com/acubeLab/k2tree}{\texttt{K2-EDF-1} \& \texttt{K2-CEDF}: https://github.com/acubeLab/k2tree}}: implementations based on the depth-first representations described in Sections~\ref{subsec:enriched-df} and~\ref{subsec:comp-df}, respectively. For the enriched representations we used six thresholds ($m_{\tau} \sqrt{N}$), where $N$ is the amount of nodes in the $k^2$-tree) for the number of nodes with $m_{\tau}$ equal to: $0.5$, $0.2$, $0.1$, $0.05$, $0.02$, and $0.01$. \texttt{K2-EDF-1} additionally compresses matrices full of $1$’s; we tested the version without this compression, but \texttt{K2-EDF-1} was always better in space and slightly faster. For both implementations, in matrix-matrix multiplication, we tested two versions: \texttt{SUB}, where, when we reach a subtree without enriched nodes, we traverse the tree to retrieve the information, and \texttt{SUB + DYN}, where, when reaching a subtree without enriched nodes, we compute the enriched values dynamically. This implementation was done in C.
    
    \item \texttt{K2-BP} and \texttt{K2-CBP}\footnote{\href{https://github.com/Yhatoh/k2tree}{\texttt{K2-BP} \& \texttt{K2-CBP}: https://github.com/Yhatoh/k2tree}}: implementations in C based on the balanced-parenthesis representations described in Sections~\ref{subsec:BPeff} and~\ref{subsec:comp-bp}. We implemented them by using: (1) $0.2\cdot \sqrt{N}$ as threshold for the subtree size of enriched nodes, where $N$ is the amount of nodes in the tree; (2) $e=256$ for the size of the blocks of the excess information; (3) $b=4096$ for the rank data structure to solve $\text{\texttt{rank}}_{\text{\texttt{(()())}}}$. Both solutions were excluded from the Database dataset because the matrices are considerably larger and much sparser, making these representations non-competitive under any of the three evaluation criteria used in our experiments (i.e., disk space usage, peak memory usage, and execution time).
\end{itemize}

We did not test the solution described in Section~\ref{subsec:plain-df} because it is essentially equivalent to \texttt{K2-EDF-1} with a threshold equal to $N$, where $N$ is the number of nodes in the $k^2$-tree, meaning that no skip values are generated. As a result, it would not provide any speed advantage and would only save a negligible amount of space compared to using $m_\tau = 0.5$. 

Similarly, we did not evaluate the solution from Section~\ref{subsec:BPclas}, which was tested in the conference version of this work~\cite{conf_version}, because it was found that the new implementation of \texttt{K2-BP} required less space and was faster on the same matrices.

% ------------------------------------------
\subsection{Disk space usage}
% ------------------------------------------

Starting with the Webgraph dataset,~\Cref{tab:webgraph_disc_space} reports the disk space usage, in bits per nonzero, for each matrix. First, \texttt{K2-BP} uses the largest number of bits per nonzero in every case. This is expected based on the theoretical size of the BP representation on $k^2$-trees. However, we observe that \texttt{K2-CBP} can considerably reduce space usage through subtree compression. This shows that even if the values in the submatrices differ, the number of repeated subtrees remains significant. It even uses less space than \texttt{K2-TREE} in three datasets: \texttt{ID-2004}, \texttt{AR-2005}, and \texttt{UK-2005}, which are the largest matrices. Larger matrices produce larger $k^2$-trees and therefore increase the chances of achieving more compression. For the remaining matrices, \texttt{K2-CBP} uses more space than other solutions except for \texttt{IN-2004}, where it uses less space than \texttt{K2-DFUDS}. This is because smaller matrices make it harder to achieve sufficient compression to offset the two-times overhead of the BP representation.

\begin{table}[t]
    \centering 
    \caption{Table reports the disk space usage in bits of each experimented data structure per nonzero element of each matrix in the Webgraph dataset. The number in the second column after \texttt{K2-EDF-1} and \texttt{K2-CEDF} represents the different values of $m_\tau$, meaning that the respective solution will compute skip values for subtrees larger than $m_\tau\cdot \sqrt{N}$, where $N$ is the number of nodes in the respective $k^2$-tree. The matrices are ordered by matrix size from left to right.\label{tab:webgraph_disc_space}}
    \begin{tabular}{l l r r r r r r r}
    \toprule
    Solution &  & \texttt{CN-2000} & \texttt{AM-2008} & \texttt{EU-2005} & \texttt{IN-2004} & \texttt{ID-2004} & \texttt{AR-2005} & \texttt{UK-2005} \\
    \midrule
    \texttt{K2-TREE} &  & 3.72 & 9.88 & 4.33 & 3.11 & 2.56 & 2.92 & 2.89 \\
    \midrule
    \texttt{K2-DFUDS} &  & 4.40 & 11.68 & 5.12 & 3.68 & 3.03 & 3.46 & 3.42 \\
    \midrule
    \multirow{6}{*}{\texttt{K2-EDF-1}} & 0.5 & 3.05 & \textbf{9.31} & 3.65 & 2.35 & 1.62 & 2.09 & 2.02 \\
     & 0.2 & 3.07 & 9.35 & 3.66 & 2.37 & 1.63 & 2.09 & 2.02 \\
     & 0.1 & 3.12 & 9.42 & 3.68 & 2.39 & 1.63 & 2.10 & 2.02 \\
     & 0.05 & 3.20 & 9.56 & 3.72 & 2.42 & 1.64 & 2.10 & 2.03 \\
     & 0.02 & 3.41 & 9.90 & 3.83 & 2.51 & 1.67 & 2.12 & 2.04 \\
     & 0.01 & 3.69 & 10.38 & 3.99 & 2.65 & 1.70 & 2.15 & 2.06 \\
    \midrule
    \multirow{6}{*}{\texttt{K2-CEDF}} & 0.5 & \textbf{2.99} & 10.38 & \textbf{3.52} & \textbf{2.22} & \textbf{1.44} & \textbf{1.85} & \textbf{1.80} \\
     & 0.2 & 3.02 & 10.42 & 3.53 & 2.23 & \textbf{1.44} & \textbf{1.85} & \textbf{1.80} \\
     & 0.1 & 3.06 & 10.49 & 3.55 & 2.24 & \textbf{1.44} & \textbf{1.85} & \textbf{1.80} \\
     & 0.05 & 3.13 & 10.62 & 3.58 & 2.27 & 1.45 & \textbf{1.85} & 1.81 \\
     & 0.02 & 3.33 & 10.96 & 3.68 & 2.36 & 1.47 & 1.87 & 1.82 \\
     & 0.01 & 3.60 & 11.43 & 3.82 & 2.47 & 1.50 & 1.89 & 1.83 \\
    \midrule
    \texttt{K2-BP} &  & 7.53 & 20.08 & 8.03 & 5.51 & 3.99 & 4.72 & 4.66 \\
    \midrule
    \texttt{K2-CBP} &  & 5.16 & 14.86 & 5.38 & 3.50 & 2.17 & 2.70 & 2.63 \\
    \bottomrule
  \end{tabular}
\end{table}

\texttt{K2-DFUDS} consistently uses more space than \texttt{K2-TREE}, \texttt{K2-EDF-1}, and \texttt{K2-CEDF}, which is expected from its theoretical space requirements.

For \texttt{K2-EDF-1} and \texttt{K2-CEDF}, we observe that for different values of $\tau$, the space used by skip values does not significantly affect the total space in larger matrices such as \texttt{ID-2004}, \texttt{AR-2005}, and \texttt{UK-2005}. \texttt{K2-CEDF} achieves the best disk usage for every matrix except \texttt{AM-2008}. This indicates that there are compressible subtrees beyond matrices full of $1$’s, and that these contribute significantly to the compression compared to \texttt{K2-EDF-1}. At the same time, \texttt{K2-EDF-1} shows that there is also a considerable number of matrices full of $1$’s. In the case of \texttt{AM-2008}, where \texttt{K2-CEDF} uses more space, this can be explained by a small number of compressible subtrees, making the overhead of auxiliary data structures too large to outperform \texttt{K2-EDF-1}, although it still uses less space than \texttt{K2-TREE}.

\begin{table}[t]
    \centering
    \caption{Table reports the disk space usage in bits of each experimented data structure per nonzero element of each matrix in the Database dataset. The number in the second column after \texttt{K2-EDF-1} and \texttt{K2-CEDF} represents the different values of $m_\tau$, meaning that the respective solution will compute skip values for subtrees larger than $m_\tau\cdot \sqrt{N}$, where $N$ is the number of nodes in the respective $k^2$-tree. For the purpose of readability, \texttt{K2-BP} and \texttt{K2-CBP} were excluded because they were not competitive at all in the Database dataset. The matrices are ordered by how were paired for matrix-matrix operations ($M_1 + M_2$ and $M_1 \cdot M_2$), begin top table $M_1$ and bottom table $M_2$. \label{tab:database_disc_space}}
    \footnotesize
    \begin{tabular}{l l r r r r r r r r r}
    \toprule
    Solution &  & \texttt{0001} & \texttt{0004} & \texttt{0006} & \texttt{0007} & \texttt{0008} & \texttt{0012} & \texttt{0159} & \texttt{0606} & \texttt{1619} \\
    \midrule
    \texttt{K2-TREE} &  & 11.93 & 37.39 & 27.56 & 15.77 & 43.79 & 18.99 & 19.30 & 44.93 & 42.50 \\
    \midrule
    \texttt{K2-DFUDS} &  & 14.11 & 44.22 & 32.60 & 18.65 & 51.79 & 22.46 & 22.82 & 53.14 & 50.26 \\
    \midrule
    \multirow{6}{*}{\texttt{K2-EDF-1}} & 0.5 & 11.22 & 35.17 & 25.93 & 14.83 & 41.20 & 17.87 & 18.17 & 42.27 & 39.99 \\
     & 0.2 & 11.23 & 35.20 & 25.95 & 14.85 & 41.22 & 17.89 & 18.21 & 42.31 & 40.03 \\
     & 0.1 & 11.24 & 35.23 & 25.98 & 14.87 & 41.27 & 17.91 & 18.28 & 42.38 & 40.10 \\
     & 0.05 & 11.26 & 35.31 & 26.03 & 14.91 & 41.37 & 17.97 & 18.39 & 42.52 & 40.23 \\
     & 0.02 & 11.31 & 35.53 & 26.20 & 15.03 & 41.65 & 18.11 & 18.68 & 42.93 & 40.62 \\
     & 0.01 & 11.39 & 35.87 & 26.48 & 15.24 & 42.11 & 18.31 & 19.08 & 43.57 & 41.22 \\
    \midrule
    \multirow{6}{*}{\texttt{K2-CEDF}} & 0.5 & \textbf{5.74} & \textbf{31.16} & \textbf{20.60} & \textbf{10.02} & \textbf{38.00} & \textbf{15.92} & \textbf{17.61} & \textbf{40.25} & \textbf{38.25} \\
     & 0.2 & \textbf{5.74} & 31.17 & 20.61 & \textbf{10.02} & 38.01 & 15.93 & 17.64 & 40.28 & 38.28 \\
     & 0.1 & \textbf{5.74} & 31.20 & 20.63 & 10.03 & 38.04 & 15.94 & 17.69 & 40.32 & 38.32 \\
     & 0.05 & 5.75 & 31.24 & 20.66 & 10.05 & 38.10 & 15.98 & 17.79 & 40.41 & 38.40 \\
     & 0.02 & 5.77 & 31.38 & 20.76 & 10.11 & 38.28 & 16.08 & 18.01 & 40.67 & 38.64 \\
     & 0.01 & 5.80 & 31.60 & 20.91 & 10.21 & 38.56 & 16.22 & 18.34 & 41.09 & 39.03 \\
    \bottomrule
    
     &  &  &  &  &  &  &  &  &  & \\
    \toprule
    Solution &  & \texttt{2831} & \texttt{3744} & \texttt{0889} & \texttt{3868} & \texttt{2670} & \texttt{3936} & \texttt{4660} & \texttt{3867} & \texttt{3935} \\
    \midrule
    \texttt{K2-TREE} &  & 37.97 & 53.25 & 49.02 & 45.00 & 43.59 & 28.54 & 47.69 & 44.84 & 36.50 \\
    \midrule
    \texttt{K2-DFUDS} &  & 44.91 & 62.98 & 57.97 & 53.22 & 51.56 & 33.76 & 56.41 & 53.04 & 43.17 \\
    \midrule
    \multirow{6}{*}{\texttt{K2-EDF-1}} & 0.5 & 35.71 & 50.11 & 46.12 & 42.34 & 41.00 & 26.86 & 44.88 & 42.20 & \textbf{34.34} \\
     & 0.2 & 35.73 & 50.15 & 46.17 & 42.39 & 41.02 & 26.90 & 44.92 & 42.24 & 34.38 \\
     & 0.1 & 35.75 & 50.23 & 46.24 & 42.46 & 41.06 & 26.96 & 44.99 & 42.31 & 34.44 \\
     & 0.05 & 35.79 & 50.38 & 46.39 & 42.61 & 41.13 & 27.07 & 45.13 & 42.46 & 34.57 \\
     & 0.02 & 35.92 & 50.82 & 46.81 & 43.04 & 41.33 & 27.37 & 45.55 & 42.89 & 34.95 \\
     & 0.01 & 36.12 & 51.55 & 47.48 & 43.75 & 41.66 & 27.78 & 46.24 & 43.60 & 35.54 \\
    \midrule
    \multirow{6}{*}{\texttt{K2-CEDF}} & 0.5 & \textbf{32.91} & \textbf{45.72} & \textbf{42.78} & \textbf{39.36} & \textbf{37.87} & \textbf{25.76} & \textbf{41.59} & \textbf{39.20} & 34.52 \\
     & 0.2 & 32.92 & 45.75 & 42.80 & 39.38 & 37.88 & 25.78 & 41.61 & 39.23 & 34.55 \\
     & 0.1 & 32.93 & 45.79 & 42.85 & 39.43 & 37.90 & 25.82 & 41.66 & 39.27 & 34.59 \\
     & 0.05 & 32.96 & 45.87 & 42.94 & 39.51 & 37.94 & 25.91 & 41.74 & 39.36 & 34.69 \\
     & 0.02 & 33.03 & 46.10 & 43.19 & 39.78 & 38.06 & 26.13 & 42.00 & 39.62 & 34.97 \\
     & 0.01 & 33.15 & 46.48 & 43.60 & 40.21 & 38.26 & 26.45 & 42.40 & 40.05 & 35.44 \\
    \bottomrule
  \end{tabular}
\end{table}

Moving to the Database dataset,~\Cref{tab:database_disc_space} reports the disk space usage, in bits per nonzero entry, for each matrix. We observe that \texttt{K2-EDF-1} and \texttt{K2-CEDF} use the least space across all matrices, as in the Webgraph dataset. In particular, \texttt{K2-CEDF} achieves the best compression except for matrix \texttt{3935}, where \texttt{K2-EDF-1} uses slightly less space. This suggests that a considerable portion of the subtrees compressed by \texttt{K2-CEDF} correspond to matrices full of $1$’s. For the remaining matrices, the number of such matrices is smaller, showing that although these matrices are sparse, there is still a high level of repetition, resulting in repeated $k^2$-trees.

\begin{table}[t]
    \centering 
    \caption{Table reports the disk space usage in bits of each experimented data structure per nonzero element of each matrix in the Random dataset. The number in the second column after \texttt{K2-EDF-1} and \texttt{K2-CEDF} represents the different values of $m_\tau$, meaning that the respective solution will compute skip values for subtrees larger than $m_\tau\cdot \sqrt{N}$, where $N$ is the number of nodes in the respective $k^2$-tree. The matrices are ordered from left to right by density.\label{tab:random_disc_space}}
    \begin{tabular}{l l r r r r r}
    \toprule
    Solution &  & $10^{-4}$ & $10^{-3}$ & $10^{-2}$ & $10^{-1}$ & $2\cdot 10^{-1}$ \\
    \midrule
    \texttt{K2-TREE} &  & 34.08 & \textbf{21.30} & 13.46 & 6.72 & 4.88 \\
    \midrule
    \texttt{K2-DFUDS} &  & 36.16 & 24.80 & 15.88 & 7.95 & 5.77 \\
    \midrule
    \multirow{6}{*}{\texttt{K2-EDF-1}} & 0.5 & \textbf{31.56} & 21.50 & \textbf{13.36} & \textbf{6.40} & \textbf{4.66} \\
     & 0.2 & 33.88 & 24.55 & 14.17 & 6.63 & 4.74 \\
     & 0.1 & 33.88 & 26.72 & 15.46 & 7.07 & 5.20 \\
     & 0.05 & 33.88 & 27.22 & 17.51 & 7.56 & 5.21 \\
     & 0.02 & 33.88 & 27.22 & 18.93 & 9.85 & 6.96 \\
     & 0.01 & 33.88 & 27.22 & 18.93 & 10.35 & 7.01 \\
    \midrule
    \multirow{6}{*}{\texttt{K2-CEDF}} & 0.5 & 38.72 & 25.26 & 15.58 & 7.23 & 5.45 \\
     & 0.2 & 41.04 & 28.10 & 16.29 & 7.46 & 5.55 \\
     & 0.1 & 41.04 & 30.42 & 17.67 & 7.88 & 6.00 \\
     & 0.05 & 41.04 & 30.96 & 19.54 & 8.39 & 6.02 \\
     & 0.02 & 41.04 & 30.96 & 21.04 & 10.67 & 7.76 \\
     & 0.01 & 41.04 & 30.96 & 21.04 & 11.18 & 7.83 \\
    \midrule
    \texttt{K2-BP} &  & 1357.12 & 199.78 & 60.57 & 15.16 & 9.99 \\
    \midrule
    \texttt{K2-CBP} &  & 1295.92 & 200.72 & 59.79 & 14.87 & 10.04 \\
    \bottomrule
    \end{tabular}
\end{table}

For the Random dataset,~\Cref{tab:random_disc_space} reports the disk space usage, in bits per nonzero entry, for each matrix. We see that \texttt{K2-BP} and \texttt{K2-CBP} require substantial space for sparse matrices. It is interesting that for densities $10^{-2}$ and $10^{-1}$, \texttt{K2-CBP} uses less space than \texttt{K2-BP}. This shows that when the leaves at the last levels, which contain the actual matrix values, are not considered, different matrices can produce identical trees. \texttt{K2-CEDF} uses more space than \texttt{K2-TREE} but less than \texttt{K2-DFUDS} for denser matrices. This can be explained by two reasons: first, random matrices make it difficult to find repetitions; second, using more space than \texttt{K2-TREE} indicates that although some identical subtrees are found, the overhead of compressing them outweighs the benefit.

On the other hand, \texttt{K2-EDF-1} uses less space for every density except $10^{-3}$. For density $10^{-4}$, the generated $k^2$-tree is small, so there is a low chance of finding matrices full of $1$’s. However, the space overhead of skip values in \texttt{K2-EDF-1} is still more efficient than the rank data structure used in \texttt{K2-TREE}. For density $10^{-3}$, \texttt{K2-TREE} uses slightly less space, as this is a boundary case where the $k^2$-tree starts to grow. In this case, skip values consume more space because more subtrees exceed the thresholds, but no matrices full of $1$’s are found, so no compression is achieved. For denser matrices, the number of matrices full of $1$’s increases, allowing \texttt{K2-EDF-1} to use slightly less space than \texttt{K2-TREE}. Clearly, smaller thresholds use more space, since the matrices are small enough that their contribution is not negligible.

To summarise the disk space results, compressing subtrees can greatly reduce the space usage of $k^2$-trees, and in many cases the final size is smaller than the canonical \texttt{K2-TREE}. This shows that many subtrees are repeated. In very sparse matrices, \texttt{K2-BP} and \texttt{K2-CBP} use much more space. Even if there are repeated subtrees, \texttt{K2-CBP} cannot take advantage of them because its 2$\times$ space overhead is too large. Still, the results confirm that large matrices contain many repeated subtrees.

In the WebGraph dataset, \texttt{K2-EDF-1} achieves strong compression. This is expected because these matrices have strong value centrality, making blocks full of $1$s more common. In the Database matrices, \texttt{K2-EDF-1} compresses less than subtree compression, but it still provides some space savings. Finally, when there are no all-$1$ matrices, the space used by the skip values in \texttt{K2-EDF-1} is smaller than the rank data structure in \texttt{K2-TREE}. The best option for both WebGraph and Database is \texttt{K2-CEDF}, which compresses well thanks to subtree repetition, showing that these matrices generally have repeated subtrees because they contain repeated submatrices.

% ------------------------------------------
\subsection{Matrix-vector multiplication}
\label{sec:mv-mul}
% ------------------------------------------

For matrix-vector multiplication,~\Cref{tab:webgraph_result_mv} shows the average running time for the WebGraph matrices. In this dataset, \texttt{K2-EDF-1} and \texttt{K2-CEDF} are consistently faster than the other solutions, with \texttt{K2-EDF-1} being the fastest overall.

\begin{table}[t]
    \centering
    \caption{Results for matrix-vector multiplication in the WebGraph dataset, reporting the average running time in seconds. The matrices are ordered by size from left to right.\label{tab:webgraph_result_mv}}
    \begin{tabular}{l r r r r r r r}
    \toprule
    Solution & \texttt{\texttt{CN-2000}} & \texttt{\texttt{AM-2008}} & \texttt{\texttt{EU-2005}} & \texttt{\texttt{IN-2004}} & \texttt{\texttt{ID-2004}} & \texttt{\texttt{AR-2005}} & \texttt{\texttt{UK-2005}} \\
     & \textit{time} & \textit{time} & \textit{time} & \textit{time} & \textit{time} & \textit{time} & \textit{time} \\
    \midrule
    \texttt{K2-TREE} & 0.11 & 0.50 & 0.76 & 0.46 & 3.91 & 16.49 & 23.93 \\
    \midrule
    \texttt{K2-DFUDS} & 0.14 & 0.54 & 0.92 & 0.66 & 6.30 & 23.01 & 31.12 \\
    \midrule
    \texttt{K2-EDF-1} & \textbf{0.04} & \textbf{0.24} & \textbf{0.31} & \textbf{0.17} & \textbf{1.29} & \textbf{5.60} & \textbf{8.54} \\
    \midrule
    \texttt{K2-CEDF} & 0.05 & 0.25 & 0.37 & 0.21 & 1.82 & 7.87 & 13.00 \\
    \midrule
    \texttt{K2-BP} & 0.12 & 0.63 & 0.88 & 0.53 & 4.47 & 17.70 & 25.89 \\
    \midrule
    \texttt{K2-CBP} & 0.23 & 1.58 & 1.93 & 1.05 & 9.86 & 40.72 & 67.71 \\
    \bottomrule
    \end{tabular}
\end{table}

For \texttt{K2-CEDF}, the reason is that although we must move from a node to its reference, matrix-vector multiplication only requires a left-to-right scan. The jumps may be costly because they require a rank operation on the pattern \texttt{0000}, but the results show that this overhead is not large enough to make it slower than the other representations. \texttt{K2-EDF-1} is the fastest for the same reason, it also benefits from left-to-right scanning, and it additionally takes advantage of blocks full of $1$s: in those cases, the algorithm can add the value directly without performing extra traversal.

\texttt{K2-DFUDS} is slower than \texttt{K2-TREE}, which is expected. Even though it performs a left-to-right scan, it must decode extra bits in its representation, as described in Section~\ref{subsec:dfuds}.

\texttt{K2-CBP} is clearly the slowest implementation. This is expected because it stores roughly twice as many bits and uses rank operations when jumping between a pruned subtree and its reference. Finally, \texttt{K2-BP} is slower on smaller matrices, but for larger ones, left-to-right scanning helps compensate for the structure's overhead, making it more competitive as matrix size increases.

\begin{table}[t]
    \centering
    \caption{Results for matrix-vector multiplication in the Database dataset, reporting the average running time in seconds. For the purpose of readability, \texttt{K2-BP} and \texttt{K2-CBP} were excluded because they were not competitive at all in the Database dataset. The matrices are ordered according to how they were paired for the matrix–matrix operations ($M_1 + M_2$ and $M_1 \cdot M_2$). The top table corresponds to $M_1$, and the bottom table corresponds to $M_2$.\label{tab:database_result_mv}}
  \begin{tabular}{l r r r r r r r r r}
    \toprule
    Solution & \texttt{0001} & \texttt{0004} & \texttt{0006} & \texttt{0007} & \texttt{0008} & \texttt{0012} & \texttt{0159} & \texttt{0606} & \texttt{1619} \\
     & \textit{time} & \textit{time} & \textit{time} & \textit{time} & \textit{time} & \textit{time} & \textit{time} & \textit{time} & \textit{time} \\
    \midrule
    \texttt{K2-TREE} & 12.89 & 22.75 & 16.07 & 8.63 & 19.64 & 10.47 & 2.80 & 10.70 & 9.73 \\
    \midrule
    \texttt{K2-DFUDS} & 12.18 & 19.32 & 14.00 & 8.01 & 17.20 & 8.88 & 2.00 & 7.99 & 7.44 \\
    \midrule
    \texttt{K2-EDF-1} & \textbf{4.28} & \textbf{12.04} & \textbf{8.02} & \textbf{3.21} & \textbf{11.07} & \textbf{5.09} & \textbf{1.45} & \textbf{5.69} & \textbf{5.46} \\
    \midrule
    \texttt{K2-CEDF} & 6.50 & 23.79 & 15.07 & 5.51 & 21.92 & 7.75 & 1.78 & 10.62 & 9.77 \\
    \bottomrule

     & & & & & & & & & \\
    \toprule
    Solution & \texttt{2831} & \texttt{3744} & \texttt{0889} & \texttt{3868} & \texttt{2670} & \texttt{3936} & \texttt{4660} & \texttt{3867} & \texttt{3935} \\
     & \textit{time} & \textit{time} & \textit{time} & \textit{time} & \textit{time} & \textit{time} & \textit{time} & \textit{time} & \textit{time} \\
    \midrule
    \texttt{K2-TREE} & 69.40 & 11.41 & 11.00 & 8.34 & 36.89 & 5.77 & 9.77 & 8.04 & 8.05 \\
    \midrule
    \texttt{K2-DFUDS} & 60.41 & 9.37 & 9.04 & 6.52 & 31.44 & \textbf{4.21} & 8.05 & 6.44 & 6.22 \\
    \midrule
    \texttt{K2-EDF-1} & \textbf{36.00} & \textbf{6.40} & \textbf{6.04} & \textbf{4.45} & \textbf{22.13} & 4.32 & \textbf{6.80} & \textbf{5.50} & \textbf{5.72} \\
    \midrule
    \texttt{K2-CEDF} & 96.92 & 12.73 & 10.94 & 7.71 & 54.45 & 5.32 & 13.18 & 10.43 & 9.19 \\
    \bottomrule
  \end{tabular}
\end{table}

Furthermore,~\Cref{tab:database_result_mv} reports the average running time for the Database matrices. The results are similar to those from the WebGraph dataset: \texttt{K2-EDF-1} is the fastest solution in almost all cases, except for matrix \texttt{3936} (but for a small amount). This exception can be explained by looking at the disk space usage in~\Cref{tab:database_disc_space}, where the difference between \texttt{K2-TREE} and \texttt{K2-EDF-1} is very small. This indicates that for this matrix, there are almost no all-$1$ submatrices to compress, so \texttt{K2-EDF-1} does not benefit from that advantage.

Additionally,~\Cref{tab:info-database} shows that matrix \texttt{3936} has one of the fewest $1$s and a small number of nodes per non-zero. This may explain why \texttt{K2-DFUDS} performs better: in this specific case, \texttt{K2-DFUDS} may read its bits more efficiently than \texttt{K2-EDF-1}, leading to better cache behaviour.

\begin{table}[t]
    \centering
    \caption{Results for matrix-vector multiplication in the Random dataset, reporting the average running time in seconds. The matrices are ordered by density from left to right.\label{tab:random_result_mv}}
  \begin{tabular}{l r r r r r}
    \toprule
    Solution &  $10^{-4}$ & $10^{-3}$ & $10^{-2}$ & $10^{-1}$ & $2\cdot 10^{-1}$ \\
    & \textit{time} & \textit{time} & \textit{time} & \textit{time} & \textit{time} \\
    \midrule
    \texttt{K2-TREE} & 0.0080 & 0.0370 & 0.2192 & 1.0138 & 1.1155 \\
    \midrule
    \texttt{K2-DFUDS} & \textbf{0.0050} & 0.0350 & 0.2525 & 1.4800 & 2.3575 \\
    \midrule
    \texttt{K2-EDF-1} & 0.0063 & \textbf{0.0219} & \textbf{0.0957} & 0.4562 & \textbf{0.6313} \\
    \midrule
    \texttt{K2-CEDF} & 0.0065 & 0.0233 & 0.1211 & \textbf{0.4446} & 0.6734 \\
    \midrule
    \texttt{K2-BP} & 0.0080 & 0.0403 & 0.2408 & 1.1761 & 1.7663 \\
    \midrule
    \texttt{K2-CBP} & 0.0088 & 0.0451 & 0.2959 & 1.6144 & 2.1699 \\
    \bottomrule
  \end{tabular}
\end{table}

Lastly,~\Cref{tab:random_result_mv} reports the average running time for the Random matrices. The results lead to the same conclusion as before: \texttt{K2-EDF-1} and \texttt{K2-CEDF} are faster when the matrices are denser, while \texttt{K2-DFUDS} performs better when the matrices are sparser. Additionally, this figure shows that \texttt{K2-BP} and \texttt{K2-CBP} are slower on sparse matrices, but \texttt{K2-BP} can be faster than \texttt{K2-DFUDS} as the matrices become denser.

To summarise the results for matrix-vector multiplication, we observe that \texttt{K2-EDF-1} is almost always the best choice, except when the matrix is very sparse, and the tree is small. Even though \texttt{K2-CEDF} must perform a rank operation over the pattern \texttt{0000}, which can negatively impact execution time, it becomes competitive when the number of pruned subtrees is large. This is particularly evident in the Database dataset, where \texttt{K2-CEDF} takes roughly twice as long as \texttt{K2-EDF-1} to compute the multiplication. However, thanks to its depth-first layout, it benefits from left-to-right scanning, allowing it to remain competitive in the WebGraph and Random datasets.

\texttt{K2-TREE} is slower than our solutions, which is expected due to the extra rank operations it performs. However, it can still be faster when the tree is small enough, since the cost of these rank operations becomes negligible compared to the number of bits that must be read when the tree is large (i.e., when the matrix is denser). In those cases, \texttt{K2-TREE} can outperform \texttt{K2-DFUDS}.

% ------------------------------------------
\subsection{Matrix-matrix sum}
\label{sec:mm-sum}
% ------------------------------------------

Before analysing the experimental results, it is important to note that, instead of Matrix-vector multiplication, this operation constructs a $k^2$-tree. Thus, the time required can be affected not only by the operation itself but also by building the final data structure.

\begin{table}[htbp]
    \centering
    \caption{Results for matrix-matrix sum in matrices \texttt{CN-2000}, \texttt{AM-2008}, \texttt{EU-2005} and \texttt{IN-2004} from the WebGraph dataset, reporting the peak memory usage in bits per nonzero element in the matrix and running time in seconds. The matrices are ordered by size from left-right.\label{tab:webgraph_result_sum_small}}
  \begin{tabular}{l r r r r r r r r}
      \toprule
    Solution & \multicolumn{2}{c}{\texttt{CN-2000}} & \multicolumn{2}{c}{\texttt{AM-2008}} & \multicolumn{2}{c}{\texttt{EU-2005}} & \multicolumn{2}{c}{\texttt{IN-2004}} \\
    \cmidrule(lr){2-3} \cmidrule(lr){4-5} \cmidrule(lr){6-7} \cmidrule(lr){8-9}
     & \textit{mem} & \textit{time} & \textit{mem} & \textit{time} & \textit{mem} & \textit{time} & \textit{mem} & \textit{time} \\
    \midrule
    \texttt{K2-TREE} & 13.25 & 0.19 & 24.53 & 0.95 & 11.25 & 1.28 & 10.20 & 0.30 \\
    \midrule
    \texttt{K2-DFUDS} & 22.15 & \textbf{0.08} & 28.14 & \textbf{0.36} & 11.44 & \textbf{0.42} & 9.03 & \textbf{0.26} \\
    \midrule
    \texttt{K2-EDF-1} & \textbf{9.01} & \textbf{0.08} & \textbf{17.52} & 0.44 & \textbf{6.95} & 0.61 & \textbf{4.72} & 0.33 \\
    \midrule
    \texttt{K2-CEDF} & 9.35 & 0.14 & 18.50 & 0.62 & 7.20 & 0.95 & 5.05 & 0.54 \\
    \midrule
    \texttt{K2-BP} & 15.53 & 0.21 & 34.97 & 1.18 & 15.81 & 1.45 & 10.98 & 0.92 \\
    \midrule
    \texttt{K2-CBP} & 14.08 & 0.87 & 33.51 & 5.67 & 14.46 & 7.06 & 9.72 & 3.65 \\
    \bottomrule
  \end{tabular}
\end{table}

\begin{table}[htbp]
    \centering
    \caption{Results for matrix-matrix sum in \texttt{ID-2000}, \texttt{AR-2008} and \texttt{IN-2004} from the WebGraph dataset, reporting the peak memory usage in bits per nonzero element in the matrix and running time in seconds. The matrices are ordered by size from left-right.\label{tab:webgraph_result_sum_big}}
  \begin{tabular}{l r r r r r r}
    & & & & & & \\
    \toprule
    Solution & \multicolumn{2}{c}{\texttt{ID-2004}} & \multicolumn{2}{c}{\texttt{AR-2005}} & \multicolumn{2}{c}{\texttt{UK-2005}} \\
    \cmidrule(lr){2-3} \cmidrule(lr){4-5} \cmidrule(lr){6-7}
     & \textit{mem} & \textit{time} & \textit{mem} & \textit{time} & \textit{mem} & \textit{time} \\
    \cmidrule{1-7}
    \texttt{K2-TREE} & 6.22 & 6.74 & 7.38 & 25.30 & 7.12 & 40.49 \\
    \cmidrule{1-7}
    \texttt{K2-DFUDS} & 3.20 & \textbf{0.19} & 6.28 & \textbf{7.64} & 6.08 & \textbf{12.74} \\
    \cmidrule{1-7}
    \texttt{K2-EDF-1} & \textbf{1.65} & 1.17 & \textbf{3.76} & 10.97 & \textbf{3.51} & 16.32 \\
    \cmidrule{1-7}
    \texttt{K2-CEDF} & 1.94 & 2.24 & 4.07 & 19.01 & 3.91 & 29.26 \\
    \cmidrule{1-7}
    \texttt{K2-BP} & 5.02 & 1.82 & 9.50 & 30.14 & 9.26 & 45.37 \\
    \cmidrule{1-7}
    \texttt{K2-CBP} & 4.13 & 12.04 & 7.54 & 98.35 & 7.28 & 159.24 \\
    \bottomrule
  \end{tabular}
\end{table}

Table~\ref{tab:webgraph_result_sum_small} and~\ref{tab:webgraph_result_sum_big} presents the results for the WebGraph matrices. Beginning with peak memory usage, we observe that, as expected, the two best-performing representations are \texttt{K2-EDF-1} and \texttt{K2-CEDF}. For the smaller matrices (\texttt{CN-2000}, \texttt{AM-2008}, and \texttt{EU-2005}), \texttt{K2-DFUDS}, \texttt{K2-BP}, and \texttt{K2-CBP} require more memory than \texttt{K2-TREE}. The inherent size of each representation explains this (see Lemmas~\ref{lemma:bp}, \ref{lemma:dfuds}, and \ref{lemma:cbp}). For the larger matrices (\texttt{IN-2004}, \texttt{ID-2004}, \texttt{AR-2005}, and \texttt{UK-2005}), the depth-first representations become more memory-efficient than \texttt{K2-TREE}. The reason is that, under a depth-first layout, subtrees that must be inserted into the result can be copied directly, whereas \texttt{K2-TREE} requires an auxiliary queue to store intermediate results before merging. This additional structure introduces overhead and may nearly double the memory required for the resulting tree. This also accounts for the lower peak memory usage observed in \texttt{K2-EDF-1} and \texttt{K2-CEDF}. Finally, between these two latter representations, \texttt{K2-EDF-1} achieves the smallest memory footprint. Since the WebGraph matrices have the property of centrality, in this case the resulting $k^2$-tree is denser and contains more all-$1$ submatrices; in \texttt{K2-EDF-1} these matrices are compressed into a single node, whereas \texttt{K2-CEDF} stores the result without compression, producing a larger tree and therefore higher memory usage.

Regarding execution time, \texttt{K2-DFUDS}, \texttt{K2-EDF-1}, and \texttt{K2-CEDF} all outperform \texttt{K2-TREE}. This is primarily due to the fact that the depth-first layout reduces the sum operation to a simple left-to-right scan and also enables faster construction of the resulting tree. \texttt{K2-DFUDS} is consistently faster than the other two. This can be attributed to the case where an all-$0$ submatrix is added to a non-all-$0$ submatrix: the algorithm must copy the entire $k^2$-tree of the latter into the output, and \texttt{K2-DFUDS} implements this step more efficiently. Finally, \texttt{K2-BP} and \texttt{K2-CBP} are slower than the other representations, largely because they require copying more bits under the same conditions, and in the case of \texttt{K2-CBP}, additional rank operations are needed to move from a pruned subtree to its reference.

\begin{table}[t]
  \centering
  \caption{Results for matrix-matrix sum for the matrices of the Database dataset indicated in the top row. The column $M_1+M_2$ means the result is the sum of matrices $M_1$ and $M_2$. We report the peak memory usage (in bits per nonzero element in the matrix) and the running time (in seconds). For the purpose of readability, \texttt{K2-BP} and \texttt{K2-CBP} were excluded because they were not competitive at all in the Database dataset. The matrices are ordered as in Table~\ref{tab:info-database} from left to right.\label{tab:database_result_sum}}
  \footnotesize
  \begin{tabular}{l r r r r r r r r r r}
      \toprule
    Solution & \multicolumn{2}{c}{\texttt{0001}$+$\texttt{2831}} & \multicolumn{2}{c}{\texttt{0004}$+$\texttt{3744}} & \multicolumn{2}{c}{\texttt{0006}$+$\texttt{0889}} & \multicolumn{2}{c}{\texttt{0007}$+$3868} & \multicolumn{2}{c}{\texttt{0008}$+$\texttt{2670}} \\
    \cmidrule(lr){2-3} \cmidrule(lr){4-5} \cmidrule(lr){6-7} \cmidrule(lr){8-9} \cmidrule(lr){10-11}
     & \textit{mem} & \textit{time} & \textit{mem} & \textit{time} & \textit{mem} & \textit{time} & \textit{mem} & \textit{time} & \textit{mem} & \textit{time} \\
    \midrule
    \texttt{K2-TREE} & 60.11 & 6.15 & 93.70 & 3.15 & 78.63 & 3.00 & 56.30 & 2.43 & 122.22 & 10.67 \\
    \midrule
    \texttt{K2-DFUDS} & 64.68 & \textbf{2.81} & 96.62 & \textbf{2.28} & 77.57 & \textbf{1.98} & 51.25 & \textbf{1.96} & 100.57 & \textbf{6.09} \\
    \midrule
    \texttt{K2-EDF-1} & 51.24 & 48.26 & 76.11 & 17.16 & 60.98 & 13.51 & 39.99 & 7.08 & 79.52 & 44.56 \\
    \midrule
    \texttt{K2-CEDF} & \textbf{47.50} & 121.36 & \textbf{72.11} & 42.83 & \textbf{56.31} & 33.43 & \textbf{35.70} & 14.06 & \textbf{76.44} & 100.81 \\
    \midrule \\ \ctoprule{1-9}
    Solution & \multicolumn{2}{c}{\texttt{0012}$+$\texttt{3936}} & \multicolumn{2}{c}{\texttt{0159}$+$\texttt{4660}} & \multicolumn{2}{c}{\texttt{0606}$+$\texttt{3867}} & \multicolumn{2}{c}{\texttt{1619}$+$\texttt{3935}} \\
    \cmidrule(lr){2-3} \cmidrule(lr){4-5} \cmidrule(lr){6-7} \cmidrule(lr){8-9}
    & \textit{mem} & \textit{time} & \textit{mem} & \textit{time} & \textit{mem} & \textit{time} & \textit{mem} & \textit{time} \\
    \cmidrule(lr){1-9}
    \texttt{K2-TREE} & 45.69 & 1.14 & 81.05 & 0.82 & 92.19 & 0.91 & 82.72 & 1.02 \\
    \cmidrule(lr){1-9}
    \texttt{K2-DFUDS} & 50.76 & \textbf{0.78} & 89.60 & \textbf{0.40} & 106.44 & \textbf{0.59} & 93.86 & \textbf{0.72} \\
    \cmidrule(lr){1-9}
    \texttt{K2-EDF-1} & 39.59 & 7.34 & 69.34 & 8.29 & 83.16 & 9.52 & 73.27 & 8.84 \\
    \cmidrule(lr){1-9}
    \texttt{K2-CEDF} & \textbf{37.90} & 12.56 & \textbf{67.14} & 15.38 & \textbf{80.75} & 22.67 & \textbf{72.51} & 16.84 \\
    \cbottomrule{1-9}
  \end{tabular}
\end{table}

Continuing with the Database dataset, Table~\ref{tab:database_result_sum} presents the corresponding results. These results largely reinforce the previous analysis: \texttt{K2-EDF-1} and \texttt{K2-CEDF} are significantly slower than \texttt{K2-TREE} and \texttt{K2-DFUDS}. This behavior is directly related to the size of the resulting $k^2$-trees. When comparing the number of tree nodes per nonzero entry in the WebGraph matrices (\Cref{tab:info-webgraph}) with those in the Database dataset (\Cref{tab:info-database}), we observe that the $k^2$-trees derived from the Database matrices are considerably larger, even though the matrices are sparse. Consequently, whenever the algorithm must copy an entire $k^2$-tree, \texttt{K2-EDF-1} and \texttt{K2-CEDF} require substantially more time.

In terms of peak memory usage, \texttt{K2-EDF-1} and \texttt{K2-CEDF} also exhibit the lowest memory consumption, further reinforcing the behavior previously discussed.

\begin{table}[t]
    \centering
    \caption{Results for matrix-matrix sum in the Random dataset, average running time in seconds. Each column represents the results over matrices with that density. The matrices are ordered by density from left-to right.\label{tab:random_result_sum}}
  \begin{tabular}{l r r r r r}
    \toprule
    Solution & $10^{-4}$ & $10^{-3}$ & $10^{-2}$ & $10^{-1}$ & $2\cdot 10^{-1}$ \\
    \midrule
    \texttt{K2-TREE} & 0.0044 & \textbf{0.0046} & \textbf{0.0059} & 0.0162 & 0.0203 \\
    \midrule
    \texttt{K2-DFUDS} & 0.0549 & 0.0592 & 0.0252 & 0.0237 & 0.0189 \\
    \midrule
    \texttt{K2-EDF-1} & 0.0049 & 0.0051 & 0.0065 & \textbf{0.0149} & \textbf{0.0166} \\
    \midrule
    \texttt{K2-CEDF} & 0.0050 & 0.0053 & 0.0073 & 0.0184 & 0.0227 \\
    \midrule
    \texttt{K2-BP} & 0.0039 & \textbf{0.0046} & 0.0114 & 0.0394 & 0.0509 \\
    \midrule
    \texttt{K2-CBP} & \textbf{0.0037} & 0.0048 & 0.0125 & 0.0515 & 0.0626 \\
    \bottomrule
  \end{tabular}
\end{table}

Finally, for the Random dataset,~\Cref{tab:random_result_sum} shows the results. Surprisingly, \texttt{K2-BP} and \texttt{K2-CBP} are faster than the other solutions for the sparser matrices. This happens because these matrices contain many submatrices full of $0$s. In such sparse cases, a fast subtree copy operation is important, because there are high chances to fell in the case where an all-$0$ submatrix is added to a non-all-$0$ submatrix. The implementation benefits from the excess array techniques, which lets us copy groups of bits while still allowing traversal without decompressing them bit by bit. But also the matrices are also smaller, so the number of bits to copy is not large.

Apart from this specific situation, we see that \texttt{K2-TREE} is faster for sparser matrices, while \texttt{K2-EDF-1} is faster for denser matrices, with \texttt{K2-EDF-1} being surely the best overall. This is because, in denser matrices, a fast copy function is less important, and the operations require traversing the whole tree to reach the last level and apply the sum operation.

To summarise the results for this operation, \texttt{K2-EDF-1} and \texttt{K2-CEDF} achieve the smallest peak memory usage among all solutions. This is due to the depth-first layout, which allows subtrees to be copied directly into the result while using fewer bits. However, in terms of running time, both \texttt{K2-EDF-1} and \texttt{K2-CEDF} can become slower when the algorithm must copy the full structure of a subtree, for example, when summing an all-$0$ matrix with a matrix that is not all-$0$. As a result, these methods perform better on denser matrices. \texttt{K2-DFUDS} is the fastest solution when the matrices are large and sparse, due to its more efficient implementation of the subtree-copying step mentioned above. In contrast, \texttt{K2-TREE} can be faster when the matrices are smaller and sparser.

% ------------------------------------------
\subsection{Matrix-matrix multiplication}
\label{sec:mm-mul}
% ------------------------------------------

Before analyzing the experimental results, it is important to recall the behavior discussed at the beginning of Section~\ref{sec:mm-sum}: in this operation, performance depends not only on the running time of the algorithm but also on the cost of building the resulting $k^2$-tree.

\begin{table}[t]
    \centering
    \caption{Results for matrix-matrix multiplication for matrices \texttt{CN-2000}, \texttt{AM-2008}, \texttt{EU-2005} and \texttt{IN-2004} from the WebGraph dataset, reporting the peak memory usage in bits per nonzero element and the running time in seconds. The number in the second column after \texttt{K2-EDF-1} and \texttt{K2-CEDF} represents the different values of $m_\tau$, meaning that the respective solution will compute skip values for subtrees larger than $m_\tau\cdot \sqrt{N}$, where $N$ is the number of nodes in the respective $k^2$-tree. The matrices are ordered by matrix size from left to right.\label{tab:webgraph_result_mul_small}}
    \footnotesize
  \begin{tabular}{l l l r r r r r r r r}
      \toprule
    Solution &  &  & \multicolumn{2}{c}{\texttt{CN-2000}} & \multicolumn{2}{c}{\texttt{AM-2008}} & \multicolumn{2}{c}{\texttt{EU-2005}} & \multicolumn{2}{c}{\texttt{IN-2004}} \\
    \cmidrule(lr){4-5} \cmidrule(lr){6-7} \cmidrule(lr){8-9} \cmidrule(lr){10-11}
     &  &  & \textit{mem} & \textit{time} & \textit{mem} & \textit{time} & \textit{mem} & \textit{time} & \textit{mem} & \textit{time} \\
    \midrule
    \texttt{K2-TREE} &  &  & 49.07 & 0.14 & 61.71 & \textbf{1.26} & 58.45 & 2.53 & 40.65 & 1.12 \\
    \midrule
    \texttt{K2-DFUDS} &  &  & 61.19 & \textbf{0.09} & 72.26 & 1.29 & 68.93 & \textbf{1.91} & 48.50 & \textbf{0.55} \\
    \midrule
    \multirow{12}{*}{\texttt{K2-EDF-1}} & \multirow{2}{*}{0.5} & \texttt{SUB} & \textbf{28.09} & 0.20 & \textbf{41.00} & 6.32 & 37.81 & 6.32 & 23.82 & 0.94 \\
     &  & \texttt{SUB + DYN} & 29.51 & 0.15 & 41.53 & 2.90 & \textbf{37.63} & 3.67 & 23.97 & 0.66 \\
     & \multirow{2}{*}{0.2} & \texttt{SUB} & 28.29 & 0.20 & 41.26 & 5.76 & 37.91 & 5.76 & 23.92 & 0.89 \\
     &  & \texttt{SUB + DYN} & 28.67 & 0.16 & 41.38 & 3.14 & 37.68 & 3.74 & \textbf{22.92} & 0.68 \\
     & \multirow{2}{*}{0.1} & \texttt{SUB} & 28.64 & 0.18 & 41.80 & 5.31 & 38.08 & 5.26 & 24.05 & 0.83 \\
     &  & \texttt{SUB + DYN} & 31.24 & 0.15 & 42.03 & 3.29 & 37.66 & 3.74 & 23.09 & 0.67 \\
     & \multirow{2}{*}{0.05} & \texttt{SUB} & 29.28 & 0.17 & 42.92 & 4.73 & 38.41 & 4.82 & 24.33 & 0.79 \\
     &  & \texttt{SUB + DYN} & 30.31 & 0.15 & 43.08 & 3.47 & 37.97 & 3.72 & 23.37 & 0.67 \\
     & \multirow{2}{*}{0.02} & \texttt{SUB} & 30.94 & 0.15 & 45.52 & 3.77 & 39.25 & 4.25 & 25.06 & 0.73 \\
     &  & \texttt{SUB + DYN} & 30.96 & 0.15 & 46.11 & 3.36 & 38.82 & 3.64 & 25.11 & 0.66 \\
     & \multirow{2}{*}{0.01} & \texttt{SUB} & 33.18 & 0.15 & 49.28 & 3.12 & 40.53 & 3.94 & 26.15 & 0.70 \\
     &  & \texttt{SUB + DYN} & 33.30 & 0.15 & 49.22 & 3.11 & 40.47 & 3.71 & 25.18 & 0.68 \\
    \midrule
    \multirow{12}{*}{\texttt{K2-CEDF}} & \multirow{2}{*}{0.5} & \texttt{SUB} & 32.80 & 0.21 & 42.49 & 6.22 & 37.86 & 5.98 & 27.47 & 1.04 \\
     &  & \texttt{SUB + DYN} & 35.22 & 0.16 & 42.73 & 3.08 & 38.96 & 3.90 & 27.85 & 0.82 \\
     & \multirow{2}{*}{0.2} & \texttt{SUB} & 32.72 & 0.21 & 42.89 & 5.73 & 37.92 & 5.57 & 27.54 & 1.00 \\
     &  & \texttt{SUB + DYN} & 32.59 & 0.17 & 43.13 & 3.27 & 37.64 & 3.97 & 27.69 & 0.84 \\
     & \multirow{2}{*}{0.1} & \texttt{SUB} & 33.06 & 0.19 & 43.33 & 5.28 & 38.09 & 5.11 & 27.67 & 0.96 \\
     &  & \texttt{SUB + DYN} & 33.96 & 0.17 & 42.52 & 3.43 & 37.87 & 3.97 & 28.06 & 0.83 \\
     & \multirow{2}{*}{0.05} & \texttt{SUB} & 33.53 & 0.18 & 44.44 & 4.76 & 38.37 & 4.77 & 27.88 & 0.92 \\
     &  & \texttt{SUB + DYN} & 33.65 & 0.17 & 44.03 & 3.60 & 38.44 & 3.96 & 28.32 & 0.83 \\
     & \multirow{2}{*}{0.02} & \texttt{SUB} & 35.09 & 0.16 & 47.05 & 3.89 & 39.11 & 4.29 & 28.52 & 0.86 \\
     &  & \texttt{SUB + DYN} & 35.11 & 0.16 & 47.33 & 3.50 & 39.31 & 3.89 & 28.54 & 0.83 \\
     & \multirow{2}{*}{0.01} & \texttt{SUB} & 37.19 & 0.16 & 50.72 & 3.30 & 40.20 & 4.08 & 29.42 & 0.85 \\
     &  & \texttt{SUB + DYN} & 37.19 & 0.17 & 50.48 & 3.30 & 40.20 & 3.93 & 29.44 & 0.85 \\
    \midrule
    \texttt{K2-BP} &  &  & 78.13 & 0.54 & 248.92 & 13.35 & 134.60 & 13.93 & 74.85 & 2.70 \\
    \midrule
    \texttt{K2-CBP} &  &  & 83.27 & 1.03 & 226.38 & 26.45 & 132.87 & 26.82 & 75.39 & 5.05 \\
    \bottomrule
  \end{tabular}
\end{table}
\begin{table}[htbp]
    \centering
    \caption{Results for matrix-matrix multiplication for matrices \texttt{ID-2000}, \texttt{AR-2008} and \texttt{UK-2004} from the WebGraph dataset, reporting the peak memory usage in bits per nonzero element and the running time in seconds. The number in the second column after \texttt{K2-EDF-1} and \texttt{K2-CEDF} represents the different values of $m_\tau$, meaning that the respective solution will compute skip values for subtrees larger than $m_\tau\cdot \sqrt{N}$, where $N$ is the number of nodes in the respective $k^2$-tree. The matrices are ordered by matrix size from left to right.\label{tab:webgraph_result_mul_big}}
  \begin{tabular}{l l l r r r r r r}
    \toprule
    Solution &  &  & \multicolumn{2}{c}{\texttt{ID-2004}} & \multicolumn{2}{c}{\texttt{AR-2005}} & \multicolumn{2}{c}{\texttt{UK-2005}} \\
    \cmidrule(lr){4-5} \cmidrule(lr){6-7} \cmidrule(lr){8-9}
     &  &  & \textit{mem} & \textit{time} & \textit{mem} & \textit{time} & \textit{mem} & \textit{time} \\
    \midrule
    \texttt{K2-TREE} &  &  & 28.37 & 124.83 & 37.32 & 84.11 & 25.93 & 88.87 \\
    \midrule
    \texttt{K2-DFUDS} &  &  & 33.80 & 38.91 & 44.84 & \textbf{40.96} & 29.98 & \textbf{68.52} \\
    \midrule
    \multirow{12}{*}{\texttt{K2-EDF-1}} & \multirow{2}{*}{0.5} & \texttt{SUB} & 16.17 & 14.25 & \textbf{20.79} & 104.05 & 12.50 & 332.43 \\
     &  & \texttt{SUB + DYN} & \textbf{16.12} & 9.07 & \textbf{20.79} & 53.47 & \textbf{12.44} & 127.37 \\
     & \multirow{2}{*}{0.2} & \texttt{SUB} & 16.19 & 13.66 & 20.81 & 94.08 & 12.51 & 293.71 \\
     &  & \texttt{SUB + DYN} & 16.18 & 9.37 & 20.86 & 53.55 & \textbf{12.44} & 126.44 \\
     & \multirow{2}{*}{0.1} & \texttt{SUB} & 16.23 & 12.99 & 20.83 & 85.75 & 12.53 & 282.70 \\
     &  & \texttt{SUB + DYN} & 16.24 & 9.42 & 20.89 & 54.05 & 12.47 & 130.32 \\
     & \multirow{2}{*}{0.05} & \texttt{SUB} & 16.30 & 12.27 & 20.88 & 77.84 & 12.57 & 265.31 \\
     &  & \texttt{SUB + DYN} & 16.29 & 9.44 & 20.91 & 53.65 & 12.50 & 141.55 \\
     & \multirow{2}{*}{0.02} & \texttt{SUB} & 16.50 & 10.72 & 21.01 & 71.12 & 12.68 & 237.86 \\
     &  & \texttt{SUB + DYN} & 16.51 & \textbf{9.03} & 21.02 & 53.50 & 12.63 & 151.51 \\
     & \multirow{2}{*}{0.01} & \texttt{SUB} & 16.79 & 10.63 & 21.21 & 65.13 & 12.85 & 209.48 \\
     &  & \texttt{SUB + DYN} & 16.75 & 9.66 & 21.21 & 53.24 & 12.85 & 144.99 \\
    \midrule
    \multirow{12}{*}{\texttt{K2-CEDF}} & \multirow{2}{*}{0.5} & \texttt{SUB} & 19.42 & 29.36 & 24.27 & 96.81 & 15.76 & 275.17 \\
     &  & \texttt{SUB + DYN} & 19.50 & 25.74 & 24.28 & 64.01 & 15.78 & 156.28 \\
     & \multirow{2}{*}{0.2} & \texttt{SUB} & 19.44 & 29.85 & 24.28 & 90.44 & 15.77 & 258.70 \\
     &  & \texttt{SUB + DYN} & 19.33 & 26.78 & 24.30 & 65.87 & 15.78 & 151.63 \\
     & \multirow{2}{*}{0.1} & \texttt{SUB} & 19.47 & 29.70 & 24.30 & 85.23 & 15.79 & 243.52 \\
     &  & \texttt{SUB + DYN} & 19.47 & 27.38 & 24.32 & 65.07 & 15.79 & 161.94 \\
     & \multirow{2}{*}{0.05} & \texttt{SUB} & 19.53 & 29.20 & 24.34 & 78.43 & 15.82 & 232.26 \\
     &  & \texttt{SUB + DYN} & 19.53 & 27.84 & 24.50 & 63.94 & 15.83 & 165.07 \\
     & \multirow{2}{*}{0.02} & \texttt{SUB} & 19.69 & 26.91 & 24.44 & 74.89 & 15.90 & 217.09 \\
     &  & \texttt{SUB + DYN} & 19.76 & 25.87 & 24.44 & 64.61 & 15.90 & 168.23 \\
     & \multirow{2}{*}{0.01} & \texttt{SUB} & 19.94 & 28.22 & 24.61 & 71.72 & 16.03 & 206.43 \\
     &  & \texttt{SUB + DYN} & 19.94 & 27.49 & 24.71 & 65.38 & 16.04 & 169.20 \\
    \midrule
    \texttt{K2-BP} &  &  & 53.37 & 111.32 & 77.15 & 215.50 & 64.21 & 642.67 \\
    \midrule
    \texttt{K2-CBP} &  &  & 52.13 & 146.15 & 76.34 & 391.08 & 63.04 & 1169.56 \\
    \bottomrule
  \end{tabular}
\end{table}

Tables~\ref{tab:webgraph_result_mul_small} and~\ref{tab:webgraph_result_mul_big} reports the results for the WebGraph dataset. First, \texttt{K2-BP} and \texttt{K2-CBP} are considerably slower than the other solutions. This can be explained by several factors: the number of bits that must be traversed is already twice as large as in the other representations, and when one operand is an all-$0$s matrix while the other is not, the algorithm must skip the entire second matrix. Although our representation mitigates this effect to some extent, the overhead remains substantial. Additionally, the resulting $k^2$-tree can become significantly larger than in the other solutions due to the increased bit requirements. Even so, \texttt{K2-BP} manages to be slightly faster than \texttt{K2-TREE} in \texttt{ID-2004}, suggesting that the problematic scenario described above occurs less frequently for that matrix.

\texttt{K2-DFUDS} is the fastest solution in almost all cases, except for \texttt{ID-2004}, where all variants of \texttt{K2-EDF-1} achieve better times---particularly the variant with $m_\tau = 0.2$ and dynamic skip-value computation. Across all matrices, the \texttt{SUB+DYN} variants of \texttt{K2-EDF-1} consistently outperform the \texttt{SUB}-only variants, except on small matrices and small $m_\tau$ values, where the dynamic computation introduces non-negligible overhead. For small matrices, the performance gap between \texttt{SUB} and \texttt{SUB+DYN} narrows as $m_\tau$ increases, indicating that the algorithm processes small subtrees directly rather than skipping them.

For \texttt{IN-2004}, \texttt{ID-2004}, and \texttt{AR-2005}, \texttt{K2-EDF-1} outperforms \texttt{K2-TREE}. Observing that performance differences between $m_\tau=0.01$ and $m_\tau=0.5$ (and between \texttt{SUB} and \texttt{SUB+DYN}) are small for these matrices, we conclude that the algorithm is forced to traverse most subtrees instead of skipping them. As an illustrative example, in \texttt{UK-2005} the execution time of \texttt{SUB+DYN} is nearly half that of \texttt{SUB}, showing that dynamic skip computation is particularly beneficial when large subtrees can be skipped. A similar trend holds for \texttt{K2-CEDF}, although this solution is slower than \texttt{K2-EDF-1} due to the additional rank operations required to move from a pruned subtree to its reference.

Regarding peak memory usage, \texttt{K2-EDF-1} and \texttt{K2-CEDF} are the most space-efficient across all solutions. As in the matrix--matrix sum, this is because the result is built by scanning from left to right, while \texttt{K2-TREE} requires an additional queue during merging. In contrast, \texttt{K2-DFUDS} has higher peak memory usage than \texttt{K2-TREE}, due to the larger number of bits required by the DFUDS representation and the overhead incurred during the ``conquer'' phase of the algorithm. Between the two EDF variants, \texttt{K2-CEDF} uses more space because the resulting tree is uncompressed, whereas \texttt{K2-EDF-1} compresses all-$1$s matrices. Finally, \texttt{K2-BP} and \texttt{K2-CBP} present the highest memory usage of all solutions, not only because their resulting $k^2$-trees are twice as large but also because \texttt{K2-CBP} computes auxiliary information on the fly to avoid redundant traversals, increasing memory overhead.

\begin{table}[htbp]
    \centering 
    \caption{Results for Matrix-matrix sum for \texttt{0001}, \texttt{2831}, \texttt{0004}, \texttt{3744}, \texttt{0006} and \texttt{0889} from the Database dataset, reporting the peak memory usage in bits per nonzero element in the matrix and running time in minutes. Column $M_1\times M_2$ means that the result is the multiplication of $M_1$ and $M_2$. The number in the second column after \texttt{K2-EDF-1} and \texttt{K2-CEDF} represents the different values of $m_\tau$, meaning that the respective solution will compute skip values for subtrees larger than $m_\tau\cdot \sqrt{N}$, where $N$ is the number of nodes in the respective $k^2$-tree. For the purpose of readability, \texttt{K2-BP} and \texttt{K2-CBP} were excluded because they were not competitive at all in the Database dataset. The matrices are ordered as in Table~\ref{tab:info-database}.\label{tab:database_result_mul_1}}
    \begin{tabular}{l l l r r r r r r}
    \toprule
    Solution &  &  & \multicolumn{2}{c}{\texttt{0001}$\times$\texttt{2831}} & \multicolumn{2}{c}{\texttt{0004}$\times$\texttt{3744}} & \multicolumn{2}{c}{\texttt{0006}$\times$\texttt{0889}} \\
    \cmidrule(lr){4-5} \cmidrule(lr){6-7} \cmidrule(lr){8-9}
     &  &  & \textit{mem} & \textit{time} & \textit{mem} & \textit{time} & \textit{mem} & \textit{time} \\
    \midrule
    \texttt{K2-TREE} &  &  & 70.27 & \textbf{0.03} & 54.89 & \textbf{57.30} & 87.86 & \textbf{62.68} \\
    \midrule
    \texttt{K2-DFUDS} &  &  & 83.68 & 0.06 & 66.06 & 76.81 & 95.73 & 77.31 \\
    \midrule
    \multirow{12}{*}{\texttt{K2-EDF-1}} & \multirow{2}{*}{0.5} & \texttt{SUB} & 66.22 & 1.94 & 51.79 & 856.71 & 73.69 & 702.83 \\
     &  & \texttt{SUB + DYN} & 66.22 & 0.63 & 51.80 & 208.58 & 73.70 & 192.01 \\
     & \multirow{2}{*}{0.2} & \texttt{SUB} & 66.42 & 1.15 & 52.11 & 777.31 & 73.94 & 652.07 \\
     &  & \texttt{SUB + DYN} & 66.42 & 0.50 & 52.11 & 232.53 & 73.99 & 214.77 \\
     & \multirow{2}{*}{0.1} & \texttt{SUB} & 66.76 & 0.71 & 52.61 & 596.54 & 74.38 & 668.27 \\
     &  & \texttt{SUB + DYN} & 66.75 & 0.46 & 52.61 & 244.38 & 74.39 & 230.93 \\
     & \multirow{2}{*}{0.05} & \texttt{SUB} & 67.43 & 0.39 & 53.59 & 535.46 & 75.24 & 497.44 \\
     &  & \texttt{SUB + DYN} & 67.44 & 0.25 & 53.59 & 250.78 & 75.26 & 243.95 \\
     & \multirow{2}{*}{0.02} & \texttt{SUB} & 69.37 & 0.12 & 56.51 & 402.14 & 77.75 & 496.11 \\
     &  & \texttt{SUB + DYN} & 69.37 & 0.08 & 56.51 & 246.90 & 78.36 & 246.52 \\
     & \multirow{2}{*}{0.01} & \texttt{SUB} & 72.62 & 0.14 & 61.14 & 354.70 & 81.79 & 348.52 \\
     &  & \texttt{SUB + DYN} & 72.63 & 0.11 & 61.13 & 252.13 & 81.82 & 255.85 \\
    \midrule
    \multirow{12}{*}{\texttt{K2-CEDF}} & \multirow{2}{*}{0.5} & \texttt{SUB} & \textbf{56.48} & 0.59 & \textbf{46.35} & 343.96 & \textbf{67.25} & 339.72 \\
     &  & \texttt{SUB + DYN} & 56.50 & 0.32 & 46.36 & 191.72 & 67.26 & 247.64 \\
     & \multirow{2}{*}{0.2} & \texttt{SUB} & 56.61 & 0.42 & 46.50 & 311.99 & 67.40 & 328.50 \\
     &  & \texttt{SUB + DYN} & 56.61 & 0.26 & 46.53 & 192.16 & 68.62 & 227.28 \\
     & \multirow{2}{*}{0.1} & \texttt{SUB} & 56.77 & 0.33 & 46.81 & 287.60 & 67.64 & 302.05 \\
     &  & \texttt{SUB + DYN} & 56.79 & 0.23 & 46.81 & 209.21 & 68.54 & 208.43 \\
     & \multirow{2}{*}{0.05} & \texttt{SUB} & 57.14 & 0.25 & 47.39 & 266.30 & 68.13 & 280.96 \\
     &  & \texttt{SUB + DYN} & 57.14 & 0.21 & 47.40 & 206.34 & 68.14 & 227.21 \\
     & \multirow{2}{*}{0.02} & \texttt{SUB} & 58.20 & 0.17 & 49.05 & 237.92 & 69.60 & 246.97 \\
     &  & \texttt{SUB + DYN} & 58.20 & 0.17 & 49.06 & 205.64 & 69.61 & 228.99 \\
     & \multirow{2}{*}{0.01} & \texttt{SUB} & 59.96 & 0.20 & 51.82 & 210.20 & 71.91 & 223.81 \\
     &  & \texttt{SUB + DYN} & 59.96 & 0.16 & 51.83 & 195.50 & 72.27 & 210.54 \\
    \bottomrule
  \end{tabular}
\end{table}

\begin{table}[htbp]
    \centering
    \caption{Results for Matrix-matrix sum for \texttt{0007}, \texttt{3868}, \texttt{0008}, \texttt{2670}, \texttt{0010} and \texttt{3936} from the Database dataset, reporting the peak memory usage in bits per nonzero element in the matrix and running time in minutes. Column $M_1\times M_2$ means that the result is the multiplication of $M_1$ and $M_2$. The number in the second column after \texttt{K2-EDF-1} and \texttt{K2-CEDF} represents the different values of $m_\tau$, meaning that the respective solution will compute skip values for subtrees larger than $m_\tau\cdot \sqrt{N}$, where $N$ is the number of nodes in the respective $k^2$-tree. For the purpose of readability, \texttt{K2-BP} and \texttt{K2-CBP} were excluded because they were not competitive at all in the Database dataset. The matrices are ordered as in Table~\ref{tab:info-database}.\label{tab:database_result_mul_2}}
    \begin{tabular}{l l l r r r r r r}
    
    \toprule
    Solution &  &  & \multicolumn{2}{c}{0007$\times$3868} & \multicolumn{2}{c}{0008$\times$2670} & \multicolumn{2}{c}{0012$\times$3936} \\
    \cmidrule(lr){4-5} \cmidrule(lr){6-7} \cmidrule(lr){8-9}
     &  &  & \textit{mem} & \textit{time} & \textit{mem} & \textit{time} & \textit{mem} & \textit{time} \\
    \midrule
    \texttt{K2-TREE} &  &  & 28.04 & \textbf{0.23} & 127.54 & \textbf{87.28} & 28.44 & \textbf{0.04} \\
    \midrule
    \texttt{K2-DFUDS} &  &  & 34.30 & 0.55 & 152.40 & 117.86 & 34.87 & 0.06 \\
    \midrule
    \multirow{12}{*}{\texttt{K2-EDF-1}} & \multirow{2}{*}{0.5} & \texttt{SUB} & 26.48 & 14.26 & 120.25 & 1523.39 & 26.90 & 1.02 \\
     &  & \texttt{SUB + DYN} & 26.50 & 7.05 & 120.28 & 514.18 & 26.91 & 0.42 \\
     & \multirow{2}{*}{0.2} & \texttt{SUB} & 26.71 & 6.72 & 120.79 & 1294.43 & 27.10 & 0.70 \\
     &  & \texttt{SUB + DYN} & 26.71 & 3.48 & 120.82 & 530.61 & 27.10 & 0.32 \\
     & \multirow{2}{*}{0.1} & \texttt{SUB} & 27.02 & 3.82 & 121.69 & 1063.61 & 27.48 & 0.53 \\
     &  & \texttt{SUB + DYN} & 27.02 & 2.07 & 121.71 & 476.86 & 27.49 & 0.28 \\
     & \multirow{2}{*}{0.05} & \texttt{SUB} & 27.68 & 2.24 & 123.50 & 836.30 & 28.16 & 0.39 \\
     &  & \texttt{SUB + DYN} & 27.69 & 1.29 & 123.51 & 603.15 & 28.17 & 0.23 \\
     & \multirow{2}{*}{0.02} & \texttt{SUB} & 29.53 & 1.28 & 128.83 & 634.74 & 30.05 & 0.26 \\
     &  & \texttt{SUB + DYN} & 29.54 & 1.47 & 128.83 & 399.75 & 30.04 & 0.18 \\
     & \multirow{2}{*}{0.01} & \texttt{SUB} & 32.67 & 1.19 & 137.50 & 536.66 & 32.69 & 0.31 \\
     &  & \texttt{SUB + DYN} & 32.69 & 0.95 & 137.50 & 353.67 & 32.69 & 0.26 \\
    \midrule
    \multirow{12}{*}{\texttt{K2-CEDF}} & \multirow{2}{*}{0.5} & \texttt{SUB} & \textbf{20.93} & 6.60 & \textbf{110.97} & 692.09 & \textbf{24.59} & 0.73 \\
     &  & \texttt{SUB + DYN} & 20.96 & 3.59 & 111.00 & 393.69 & 24.60 & 0.34 \\
     & \multirow{2}{*}{0.2} & \texttt{SUB} & 21.02 & 3.47 & 111.30 & 603.76 & 24.72 & 0.51 \\
     &  & \texttt{SUB + DYN} & 21.05 & 2.15 & 111.30 & 347.49 & 24.71 & 0.28 \\
     & \multirow{2}{*}{0.1} & \texttt{SUB} & 21.22 & 2.18 & 111.83 & 596.42 & 24.96 & 0.40 \\
     &  & \texttt{SUB + DYN} & 21.23 & 1.51 & 111.84 & 416.19 & 24.98 & 0.25 \\
     & \multirow{2}{*}{0.05} & \texttt{SUB} & 21.55 & 1.74 & 112.91 & 498.37 & 25.44 & 0.31 \\
     &  & \texttt{SUB + DYN} & 21.54 & 1.11 & 112.92 & 360.87 & 25.43 & 0.22 \\
     & \multirow{2}{*}{0.02} & \texttt{SUB} & 22.60 & 1.06 & 116.12 & 422.02 & 26.80 & 0.22 \\
     &  & \texttt{SUB + DYN} & 22.59 & 0.92 & 116.13 & 480.60 & 26.79 & 0.18 \\
     & \multirow{2}{*}{0.01} & \texttt{SUB} & 24.31 & 1.14 & 121.39 & 374.26 & 28.77 & 0.26 \\
     &  & \texttt{SUB + DYN} & 24.31 & 0.82 & 121.38 & 391.81 & 28.76 & 0.22 \\
    \bottomrule
  \end{tabular}
\end{table}

\begin{table}[htbp]
    \centering
    \caption{Results for Matrix-matrix sum for \texttt{0159}, \texttt{4660}, \texttt{0606}, \texttt{3867}, \texttt{1619} and \texttt{3935} from the Database dataset, reporting the peak memory usage in bits per nonzero element in the matrix and running time in minutes. Column $M_1\times M_2$ means that the result is the multiplication of $M_1$ and $M_2$. The number in the second column after \texttt{K2-EDF-1} and \texttt{K2-CEDF} represents the different values of $m_\tau$, meaning that the respective solution will compute skip values for subtrees larger than $m_\tau\cdot \sqrt{N}$, where $N$ is the number of nodes in the respective $k^2$-tree. For the purpose of readability, \texttt{K2-BP} and \texttt{K2-CBP} were excluded because they were not competitive at all in the Database dataset. The matrices are ordered as in Table~\ref{tab:info-database}.\label{tab:database_result_mul_3}}
    \begin{tabular}{l l l r r r r r r}
    \toprule
    Solution &  &  & \multicolumn{2}{c}{\texttt{0159}$\times$\texttt{4660}} & \multicolumn{2}{c}{\texttt{0606}$\times$\texttt{3867}} & \multicolumn{2}{c}{\texttt{1619}$\times$\texttt{3935}} \\
    \cmidrule(lr){4-5} \cmidrule(lr){6-7} \cmidrule(lr){8-9}
     &  &  & \textit{mem} & \textit{time} & \textit{mem} & \textit{time} & \textit{mem} & \textit{time} \\
    \midrule
    \texttt{K2-TREE} &  &  & 102.06 & \textbf{0.00} & 77.91 & \textbf{0.17} & 77.48 & \textbf{9.51} \\
    \midrule
    \texttt{K2-DFUDS} &  &  & 126.86 & \textbf{0.00} & 95.19 & 0.24 & 94.60 & 12.54 \\
    \midrule
    \multirow{12}{*}{\texttt{K2-EDF-1}} & \multirow{2}{*}{0.5} & \texttt{SUB} & 96.63 & 0.08 & 73.64 & 3.40 & 73.15 & 130.96 \\
     &  & \texttt{SUB + DYN} & 96.59 & 0.04 & 73.57 & 1.30 & 73.21 & 52.47 \\
     & \multirow{2}{*}{0.2} & \texttt{SUB} & 97.42 & 0.05 & 74.20 & 2.58 & 73.74 & 107.92 \\
     &  & \texttt{SUB + DYN} & 97.46 & 0.03 & 74.22 & 1.14 & 73.80 & 48.67 \\
     & \multirow{2}{*}{0.1} & \texttt{SUB} & 98.93 & 0.03 & 75.18 & 2.04 & 74.75 & 91.83 \\
     &  & \texttt{SUB + DYN} & 98.92 & 0.02 & 75.17 & 1.00 & 74.77 & 45.64 \\
     & \multirow{2}{*}{0.05} & \texttt{SUB} & 101.85 & 0.02 & 77.06 & 1.64 & 76.69 & 76.47 \\
     &  & \texttt{SUB + DYN} & 101.85 & 0.02 & 77.07 & 0.90 & 76.86 & 43.36 \\
     & \multirow{2}{*}{0.02} & \texttt{SUB} & 109.60 & 0.02 & 82.77 & 1.23 & 82.57 & 57.28 \\
     &  & \texttt{SUB + DYN} & 109.57 & 0.02 & 82.78 & 0.80 & 82.56 & 38.67 \\
     & \multirow{2}{*}{0.01} & \texttt{SUB} & 122.12 & 0.01 & 91.98 & 1.50 & 91.70 & 46.04 \\
     &  & \texttt{SUB + DYN} & 122.12 & 0.01 & 92.00 & 1.21 & 91.71 & 35.21 \\
    \midrule
    \multirow{12}{*}{\texttt{K2-CEDF}} & \multirow{2}{*}{0.5} & \texttt{SUB} & \textbf{90.32} & 0.05 & \textbf{69.34} & 1.66 & 71.56 & 108.45 \\
     &  & \texttt{SUB + DYN} & \textbf{90.32} & 0.04 & 69.36 & 0.86 & \textbf{71.55} & 47.42 \\
     & \multirow{2}{*}{0.2} & \texttt{SUB} & 90.73 & 0.04 & 69.68 & 1.30 & 71.93 & 91.23 \\
     &  & \texttt{SUB + DYN} & 90.73 & 0.03 & 69.65 & 0.78 & 71.88 & 47.57 \\
     & \multirow{2}{*}{0.1} & \texttt{SUB} & 91.77 & 0.03 & 70.29 & 1.14 & 72.61 & 78.59 \\
     &  & \texttt{SUB + DYN} & 91.76 & 0.03 & 70.30 & 0.74 & 72.65 & 46.14 \\
     & \multirow{2}{*}{0.05} & \texttt{SUB} & 93.65 & 0.03 & 71.47 & 0.99 & 73.95 & 68.78 \\
     &  & \texttt{SUB + DYN} & 93.65 & 0.02 & 71.43 & 0.68 & 73.96 & 44.29 \\
     & \multirow{2}{*}{0.02} & \texttt{SUB} & 98.86 & 0.03 & 74.97 & 1.20 & 77.91 & 55.51 \\
     &  & \texttt{SUB + DYN} & 98.94 & 0.03 & 74.96 & 0.62 & 77.92 & 40.63 \\
     & \multirow{2}{*}{0.01} & \texttt{SUB} & 107.03 & 0.03 & 80.74 & 1.11 & 84.35 & 44.54 \\
     &  & \texttt{SUB + DYN} & 107.05 & 0.03 & 80.78 & 1.02 & 84.38 & 37.35 \\
    \bottomrule
  \end{tabular}
\end{table}

We now continue with the results of the Database dataset, shown in Tables~\ref{tab:database_result_mul_1},~\ref{tab:database_result_mul_2} and~\ref{tab:database_result_mul_3}. As in Section~\ref{sec:mm-sum} for Matrix-matrix sum, these results reinforce the conclusions drawn from the WebGraph dataset. First, that table shows that the space usage of \texttt{K2-TREE} and the \texttt{K2-EDF-1} variants is similar, with slightly better efficiency for \texttt{K2-EDF-1} thanks to small overhead of skip values. The gap with \texttt{K2-CEDF} is much larger, implying that these matrices contain few or no all-$1$s submatrices, preventing \texttt{K2-EDF-1} from exploiting its compression.

In general, \texttt{K2-CEDF} can be considerably faster than \texttt{K2-EDF-1} on these datasets. This can be explained by the high sparsity and lack of structural correlation in the matrices: many multiplications result in subtree skips (i.e., one operand is all $0$s and the other is not). Since \texttt{K2-CEDF} achieves stronger compression, it requires fewer bits to be read than \texttt{K2-EDF-1}, resulting in fewer operations. This trend is also visible across the different values of $m_\tau$: smaller values produce more skip values and therefore enable faster subtree skipping, whereas larger $m_\tau$ forces more complete traversals without corresponding benefit.

\texttt{K2-DFUDS} is faster than both \texttt{K2-EDF-1} and \texttt{K2-CEDF} here because its range-min-max tree supports constant-time subtree skipping, giving it a significant advantage. However, \texttt{K2-TREE} outperforms \texttt{K2-DFUDS}, which can be justified by the extreme sparsity of these matrices: in such cases, constant-time rank operations are slightly more efficient than the range-min-max tree mechanism.

In terms of peak memory usage, for larger values of $m_\tau$, \texttt{K2-EDF-1} uses less memory than \texttt{K2-TREE} and \texttt{K2-DFUDS}; for smaller values of $m_\tau$, the skip values become substantial, increasing space usage. In contrast, \texttt{K2-CEDF} maintains consistently lower memory usage, except when $m_\tau = 0.01$. This can be attributed to its significantly more compact representation on disk.

\begin{table}[htbp]
    \centering
    \caption{Results for Matrix-matrix multiplication in the Random dataset, showing average running time in seconds. The number in the second column after \texttt{K2-EDF-1} and \texttt{K2-CEDF} represents the different values of $m_\tau$, meaning that the respective solution will compute skip values for subtrees larger than $m_\tau\cdot \sqrt{N}$, where $N$ is the number of nodes in the respective $k^2$-tree. The matrices are ordered by density from left to right.\label{tab:random_result_mul}}
  \begin{tabular}{l l l r r r r r}
    \toprule
    Solution &  &  & $10^{-4}$ & $10^{-3}$ & $10^{-2}$ & $10^{-1}$ & $2\cdot 10^{-1}$  \\
    \midrule
    \texttt{K2-TREE} &  &  & 0.0046 & \textbf{0.0123} & 0.2910 & 7.6887 & 17.2783 \\
    \midrule
    \texttt{K2-DFUDS} &  &  & 0.0144 & 0.0225 & 0.2694 & 5.2615 & 9.4470 \\
    \midrule
    \multirow{12}{*}{\texttt{K2-EDF-1}} & \multirow{2}{*}{0.5} & \texttt{SUB} & 0.0055 & 0.0240 & 0.4891 & 5.2221 & 6.2303 \\
     &  & \texttt{SUB + DYN} & 0.0056 & 0.0234 & 0.4268 & 4.3406 & 5.6617 \\
     & \multirow{2}{*}{0.2} & \texttt{SUB} & 0.0038 & 0.0134 & 0.4141 & 5.0345 & 5.7232 \\
     &  & \texttt{SUB + DYN} & \textbf{0.0035} & 0.0155 & 0.3901 & 4.4358 & \textbf{4.9355} \\
     & \multirow{2}{*}{0.1} & \texttt{SUB} & 0.0047 & 0.0187 & 0.4195 & 4.8169 & 6.1894 \\
     &  & \texttt{SUB + DYN} & 0.0047 & 0.0191 & 0.4476 & 4.5542 & 5.9766 \\
     & \multirow{2}{*}{0.05} & \texttt{SUB} & 0.0052 & 0.0188 & \textbf{0.2472} & 2.8633 & 6.1243 \\
     &  & \texttt{SUB + DYN} & 0.0051 & 0.0190 & 0.2600 & \textbf{2.8450} & 6.0564 \\
     & \multirow{2}{*}{0.02} & \texttt{SUB} & 0.0049 & 0.0190 & 0.3922 & 4.1838 & 5.8858 \\
     &  & \texttt{SUB + DYN} & 0.0050 & 0.0195 & 0.4092 & 4.3645 & 5.9162 \\
     & \multirow{2}{*}{0.01} & \texttt{SUB} & 0.0054 & 0.0192 & 0.3838 & 4.1888 & 5.9381 \\
     &  & \texttt{SUB + DYN} & 0.0055 & 0.0200 & 0.4074 & 4.2299 & 5.6508 \\
    \midrule
    \multirow{12}{*}{\texttt{K2-CEDF}} & \multirow{2}{*}{0.5} & \texttt{SUB} & 0.0056 & 0.0251 & 0.4961 & 5.2330 & 6.3415 \\
     &  & \texttt{SUB + DYN} & 0.0057 & 0.0244 & 0.4453 & 4.3705 & 5.5591 \\
     & \multirow{2}{*}{0.2} & \texttt{SUB} & 0.0036 & 0.0140 & 0.3890 & 5.0976 & 6.0750 \\
     &  & \texttt{SUB + DYN} & 0.0036 & 0.0144 & 0.4017 & 4.4412 & 5.4043 \\
     & \multirow{2}{*}{0.1} & \texttt{SUB} & 0.0048 & 0.0198 & 0.4484 & 4.8162 & 6.6027 \\
     &  & \texttt{SUB + DYN} & 0.0049 & 0.0203 & 0.4637 & 4.5821 & 6.5799 \\
     & \multirow{2}{*}{0.05} & \texttt{SUB} & 0.0052 & 0.0199 & 0.2822 & 2.8786 & 5.7706 \\
     &  & \texttt{SUB + DYN} & 0.0051 & 0.0202 & 0.2843 & 2.8610 & 5.6955 \\
     & \multirow{2}{*}{0.02} & \texttt{SUB} & 0.0067 & 0.0202 & 0.4045 & 4.3603 & 5.5577 \\
     &  & \texttt{SUB + DYN} & 0.0050 & 0.0207 & 0.4245 & 4.3967 & 5.5214 \\
     & \multirow{2}{*}{0.01} & \texttt{SUB} & 0.0054 & 0.0203 & 0.3962 & 4.0780 & 6.0106 \\
     &  & \texttt{SUB + DYN} & 0.0055 & 0.0213 & 0.4257 & 4.2799 & 5.9662 \\
    \midrule
    \texttt{K2-BP} &  &  & 0.0060 & 0.0867 & 2.1080 & 22.9008 & 27.4495 \\
    \midrule
    \texttt{K2-CBP} &  &  & 0.0152 & 0.2173 & 4.0551 & 40.7572 & 38.3270 \\
    \bottomrule
  \end{tabular}
\end{table}

Finally, the Random dataset (\Cref{tab:random_result_mul}) confirms the previous observations. Both \texttt{K2-BP} and \texttt{K2-CBP} are slower than \texttt{K2-TREE}. When compared to \texttt{K2-EDF-1} and \texttt{K2-CEDF}, \texttt{K2-BP} can be faster on very sparse matrices for small $m_\tau$ values, since dynamic skip computation is inexpensive for small matrices. However, for denser matrices, it becomes significantly slower. \texttt{K2-TREE} is faster than \texttt{K2-DFUDS}, \texttt{K2-EDF-1}, and \texttt{K2-CEDF} for sparse matrices, consistent with the Database dataset; but the trend reverses on dense matrices. For dense matrices, \texttt{K2-EDF-1} and \texttt{K2-CEDF} outperform the other solutions because subtree skipping becomes rare, forcing full traversals, which favor these two representations.

To summarize the results of the matrix–matrix multiplication, we observe that \texttt{K2-BP} and \texttt{K2-CBP} are consistently the slowest solutions due to their larger bitmaps and costly subtree skipping, while also exhibiting the highest peak memory usage. \texttt{K2-DFUDS} is generally the fastest approach in the WebGraph dataset thanks to constant-time subtree skipping, although \texttt{K2-EDF-1} becomes competitive or even faster in matrices where large subtrees can be skipped, especially when using dynamic skip computation. \texttt{K2-TREE} performs well on sparser matrices due to efficient rank operations, but loses its advantage on denser ones, where \texttt{K2-EDF-1} and \texttt{K2-CEDF} benefit from full traversals and reduced memory traffic. In terms of peak memory usage, \texttt{K2-EDF-1} and \texttt{K2-CEDF} are the most space-efficient, although very small $m_\tau$ values increase the overhead of skip information. Finally, across all datasets, the behavior observed in the Random dataset reaffirms these trends: sparsity favors \texttt{K2-TREE}, while density favors the EDF-based approaches.

% ------------------------------------------
\section{Conclusion and Future Work}
\label{sec:conclusion}
% ------------------------------------------

In this paper, we introduced four depth-first representations of $k^2$-trees: \texttt{K2-EDF-1} based on a plain depth-first representation and \texttt{K2-BP} based on a balanced parenthesis representation; and their respective subtree-compressed variants, namely, \texttt{K2-CEDF} and \texttt{K2-CBP}. We also proposed a linear-time algorithm, based on Suffix and LCP arrays, to efficiently detect and compress repeated subtrees.

We experimentally compared our four approaches against two established baselines: \texttt{K2-TREE}, which uses a level-wise traversal~\cite{k2-tree_vldbj}, and \texttt{K2-DFUDS}, which relies on a depth-first unary degree sequence representation~\cite{cache-friendly-boolean}. The comparison was based on three sets of binary matrices that differ in type,
size, and sparsity (WebGraph, Database, and Random) and measured disk space usage, execution time and peak memory usage (PMU) of three classical matrix operations: matrix-vector multiplication, matrix-matrix sum, and matrix-matrix multiplication. The experimental results showed that our (plain and compressed) depth-first representations are competitive and, in many cases, advantageous depending on the workload and matrix structure. In particular, our variants, \texttt{K2-BP} and \texttt{K2-CBP}, exploit the high number of repeated subtrees found in large, moderately sparse matrices, such as the largest WebGraph datasets (e.g., \texttt{ID-2005}, \texttt{AR-2005}, and \texttt{UK-2005}). 

Our implementation \texttt{K2-CEDF} consistently achieves the best disk space usage across almost all datasets. In the few cases where it does not, \texttt{K2-EDF-1} shows that skip values introduce only a small overhead compared to the rank structure used in \texttt{K2-TREE} and the range min-max tree used in \texttt{K2-DFUDS}.

In terms of peak memory consumption, our depth-first representations \texttt{K2-EDF-1} and \texttt{K2-CEDF} consistently use the least memory in every operation. This confirms that depth-first layouts can effectively exploit structural regularities in the tree to reduce memory overhead during processing. 

In terms of performance, it varies depending on both the dataset and the operation:

\begin{description}
    \item[Matrix-vector multiplication.] \texttt{K2-EDF-1} is usually the fastest. Its compact size reduces I/O costs, and it does not require additional traversal structures, unlike \texttt{K2-TREE}.

\item[Matrix-matrix sum.] \texttt{K2-DFUDS} generally performs the best. Its constant-time subtree traversal, combined with a depth-first layout, provides an advantage over \texttt{K2-TREE}.

\item[Matrix-matrix multiplication.] No single method dominates across all cases:
\begin{itemize}
    \item \texttt{K2-TREE} performs better on large and very sparse matrices due to efficient navigation.
    \item \texttt{K2-DFUDS} can outperform \texttt{K2-TREE} when the operation requires the latter to perform many rank calls to skip a subtree. This means that when multiplying two matrices, the divide-and-conquer algorithm performs just a few recursive steps where one submatrix is all zeros and the other is not.
    \item Our two best representations \texttt{K2-EDF-1} and \texttt{K2-CEDF} can be faster when the operation forces the \texttt{K2-TREE} to do many rank calls because subtree skips are frequent, or when the algorithm needs to fully traverse the tree instead of skipping subtrees. This makes them more suitable for denser matrices. Moreover, comparing these two proposals (i.e., \texttt{K2-EDF-1} and \texttt{K2-CEDF}), we observe a clear trade-off between time and space efficiency: \texttt{K2-EDF-1} favors faster query performance, while \texttt{K2-CEDF} achieves stronger compression at the cost of slightly higher execution time.
    \item As far as our two other proposals are concerned, i.e. \texttt{K2-BP} and \texttt{K2-CBP}, the former needs nearly $2$x bits compared to \texttt{K2-TREE} and \texttt{K2-EDF-1}, and the latter achieves a lower disk space usage in some cases but their execution time is generally worse than all the other approaches, except in specific cases (e.g., \texttt{K2-BP} performs better than \texttt{K2-TREE} in matrix-matrix multiplication for \texttt{ID-2004}).
\end{itemize}

\end{description}

As future work, we foresee the following research directions about our novel depth-first representations:

\begin{itemize}
    \item Study how subtree repetition can help optimize matrix operations. If there are two identical subtrees, it can mean that they have identical submatrices, so one could avoid repeating some operations that have already been executed.
    \item Implement other depth-first representations like the ultra succinct one by Jansson et al.~\cite{jansson2007ultra}; 
    \item Adapt \texttt{K2-DFUDS} in such a way that subtree compression is applied to it. As it is a depth-first layout, subtree compression is doable.
    \item Implement a parallel version of the depth-first representations shown in this paper and compare it against the parallel \texttt{K2-TREE} implementation of Arroyuelo et al.~\cite{k2-tree_vldbj}.
    \item Study how the depth-first representations showed in this paper can affect energy consumption on matrix operations, due to better use of cache access.
\end{itemize}

% ------------------------------------------
\section*{Acknowledgements}
\label{sec:ack}
% ------------------------------------------

Paolo Ferragina and Francesco Tosoni were funded by the Alfred P.~Sloan Foundation with the grant \#\href{https://sloan.org/grant-detail/g-2025-25193}{G-2025-25193} (\url{sloan.org}).

\bibliography{sn-bibliography}% common bib file

\end{document}